%% file: PT23.tex
\documentclass[runningheads,envcountsame]{llncs}
\usepackage[T1]{fontenc}

\usepackage{graphicx}

\usepackage{lineno}

\input{macros.tex}
\begin{document}
\title{On History-Deterministic One-Counter Nets}

\author{Aditya Prakash\inst{1}\orcidlink{0000-0002-2404-0707} \and
K.~S.~Thejaswini\inst{1}\orcidlink{0000-0001-6077-7514}}
\authorrunning{A. Prakash, K. S. Thejaswini}

\institute{Department of Computer Science, University of Warwick
\email{\{aditya.prakash,thejaswini.raghavan.1\}@warwick.ac.uk}}
\maketitle              
\begin{abstract}
 We consider the model of history-deterministic one-counter nets (OCNs). History-determinism is a property of transition systems that allows for a limited kind of non-determinism which can be resolved `on-the-fly'. Token games, which have been used to characterise history-determinism over various models, also characterise history-determinism over OCNs.  By reducing 1-token games to simulation games, we are able to show that checking for history-determinism of OCNs is decidable. Moreover, we prove that this problem is $\PSPACE$-complete for a unary encoding of transitions, and $\EXPSPACE$-complete for a binary encoding. 

We then study the language properties of history-deterministic OCNs. We show that the resolvers of non-determinism for history-deterministic OCNs are eventually periodic. As a consequence, for a given history-deterministic OCN, we construct a language equivalent deterministic one-counter automaton. We also show the decidability of comparing languages of history-deterministic OCNs, such as language inclusion and language universality.
\end{abstract}
\section{Introduction}
\input{intro.tex}
\section{Preliminaries}\label{sec:prelims}
\input{prelim.tex}

\section{Deciding History-Determinism}\label{sec:Decide}
\input{HDviaToken.tex}
\input{tokengames.tex}

\section{Languages and History-Determinism in OCNs}\label{sec:hdocnTodoca}
\input{rhocn.tex}

\section{Extensions and Variations of OCN}\label{sec:otherResults}
\input{miscResults.tex}

\section{Discussion}
\input{conclusion.tex}

\subsubsection{Acknowledgements} 
We would like to thank Dmitry Chistikov for listening to our conjectures and pointing us to important references. We are also grateful for his comments on our introduction. We are thankful to Neha Rino for carefully proofreading our paper, and suggesting improvements in our presentation.
We also thank Sougata Bose, Piotrek Hofman,  Filip Mazowiecki, David Purser and Patrick Totzke for their insightful remarks on our draft, and for telling us about weak simulation. We are grateful to Shaull Almagor and Asaf Yeshurun for a fun talk about OCNs. Finally, we thank Marcin Jurdzi\'nski for his support, and for bringing us his homemade rhubarb crumble.  

\bibliographystyle{splncs04}
\bibliography{hoca} 

\newpage
\appendix

\section{Appendix for Section~\ref{sec:Decide}}
\input{Appendix3.tex}
\section{Appendix for Section~\ref{sec:hdocnTodoca}}
\input{Appendix4.tex}
\section{Appendix for Section~\ref{sec:otherResults}}
\input{Appendix5.tex}
\end{document}

%% file: macros.tex
\usepackage{xcolor}
\usepackage{amsfonts}
\usepackage{amsmath}
\usepackage{bm}
\usepackage{amssymb}
\usepackage{tikzlings}
\usepackage{tikzducks}
\usepackage{tikzsymbols}
\usepackage{tikz}
\usepackage{enumitem}
\usepackage{mathtools}
\usepackage{thmtools}
\usepackage{thm-restate}

\usepackage[new]{old-arrows}

\usepackage[bold,full]{complexity}

\usepackage{orcidlink}
\usepackage{hyperref}
\usepackage{cleveref}

\newcommand{\eve}{\pmb{\exists}}
\newcommand{\adam}{\pmb{\forall}}
\newcommand{\Gc}{\mathcal{G}} 
\newcommand{\Lc}{\mathcal{L}}
\newcommand{\Ac}{\mathcal{A}}
\newcommand{\Hc}{\mathcal{H}}
\newcommand{\Bc}{\mathcal{B}}
\newcommand{\Dc}{\mathcal{D}}
\newcommand{\Mc}{\mathcal{M}}
\newcommand{\Nc}{\mathcal{N}}

\newcommand{\Sc}{\mathcal{S}}

\newcommand{\Cc}{\mathcal{C}} 
\newcommand{\ssp}{semilinear-strategy property}
\newcommand{\gut}{\text{good}}

\newcommand{\zerotest}{\mathrm{zero}}
\newcommand{\notzerotest}{\neg\mathrm{zero}}
\newcommand{\notest}{\mathrm{X}}

\newcommand{\simg}[2]{\mathcal{G}(#1 \longhookrightarrow #2)}
\newcommand{\seq}[1]{\langle #1\rangle}
\newcommand{\sqr}[1]{\left[ #1 \right]}
\newcommand{\simulates}{\longhookrightarrow}

\newcommand{\one}{\mathbf{1}}

\newcommand{\win}{\mathtt{win}}
\newcommand{\twin}{\mathtt{twin}}
\newcommand{\last}{\mathtt{last}}

\newclass{\TOWER}{TOWER}
\newclass{\ACKERMANN}{ACKERMANN}
\newclass{\IIEXPTIME}{2-EXPTIME}

\usepackage{graphicx}
\usepackage{array}
\usepackage{tikz}
\usetikzlibrary{automata, arrows, positioning, decorations.markings, decorations.pathreplacing}
\usepackage{xargs} 
\usepackage[colorinlistoftodos,prependcaption,textsize=tiny]{todonotes}



%% file: intro.tex
While deterministic automata are algorithmically efficient for problems such as synthesis or for solving games, they are often much less succinct, or less expressive than their non-deterministic counterparts. The notion of history-determinism was introduced by Henzinger and Piterman~\cite{HP06} for automata over infinite words with parity acceptance conditions, as a tool to solve synthesis games efficiently. Such automata are known to compose well with games, and hence are also called good-for-games (GFG) automata~\cite{HP06,Col12}. History-deterministic automata form a robust class of models that is both algorithmically and conceptually interesting, and has been extensively studied over the recent years~\cite{HP06,Col12,BKKS13,KS15,BKS17,BK18,BKLS20,CCL22,RK22}.

The notion of history-determinism emerged independently in the setting of cost automata, that can capture all regular cost functions as opposed to their deterministic version~\cite{Col09}.
Recently, history-determinism has been studied in other quantitative settings~\cite{BL21,BL22}, as well as infinite-state systems such as pushdown automata~\cite{GJLZ21,LZ22}, Parikh automata~\cite{EGJLZ22}, and timed automata~\cite{HLT22}.     

One-counter nets are finite-state systems along with a counter that stores a non-negative integer value that can never be explicitly tested for zero. They correspond to 1-dimensional VASS, Petri nets with exactly one unbounded place, and are a subclass of one-counter automata which do not have zero tests, and hence are also a subclass of pushdown automata. They are one of the simplest infinite-state systems, and hence many problems pertaining to one-counter nets are easier than their counterparts that subsume them. 

The structure of the resolvers that resolve non-determinism on-the-fly are crucial to understand history-determinism in various models. While for automata over infinite words with parity conditions, these resolvers take the shape of deterministic parity automata~\cite{HP06}, the situation for resolvers in history-deterministic infinite-state systems is not as well understood. Indeed, the computability of such a resolver for a given history-deterministic pushdown automaton is left as an open problem in the works of
Guha, Jecker, Lehtinen and Zimmermann~\cite{GJLZ21}. For history-deterministic Parikh automata, it is still an open problem if the resolver can be given by a deterministic Parikh transducer~\cite{EGJLZ22}. Moreover, many other problems such as deciding history-determinism or even  language inclusion among history-deterministic automata are  undecidable for pushdown automata and Parikh automata~\cite{GJLZ21,LZ22,EGJLZ22}. We consider history-determinism over a well-studied class of infinite-state systems of one-counter nets, where we are able to answer positively to all of the above questions.   

The techniques we use to answer several of these questions use results and techniques from the simulation problem over one-counter nets~\cite{HMT13,HLMT16}. This is not surprising, since simulation of various models has close ties with history-determinism~\cite{HP06,HLT22}.

\paragraph*{Our Contribution}

We study history-deterministic OCNs and establish them as a class of  infinite-state systems where many problems pertaining to history-determinism are decidable. This is unlike other classes of history-deterministic infinite-state systems that subsume them.

Firstly, we show that checking for history-determinism for a given one-counter net is $\PSPACE$-complete when the transitions are encoded in unary, and is $\EXPSPACE$-complete for a succinct encoding  (Theorem~\ref{thm:PSPACE}, Theorem~\ref{thm:succinctcomp}). We achieve the upper bound by giving a novel reduction from the 1-token game $G_1$ to the simulation problem over OCNs. 1-token games characterise history-determinism over OCNs, and thus our reduction further extends the link between history-determinism and simulation. This decidability result is in contrast to one-counter automata (OCA), where checking for history-determinism becomes undecidable by just adding zero-tests to OCNs (Theorem~\ref{thm:UndecOCA}).

Secondly, we show that resolvers for non-determinism in history-deterministic OCNs can be expressed as an eventually periodic set. Using this, we are able to determinise history-deterministic OCNs to give a language equivalent deterministic OCA. 

Finally, we show the decidability of the problems of language inclusion and language universality for history-deterministic OCNs to be in $\PSPACE$ and $\P$ respectively. This is in unlike non-deterministic OCNs, where these problems are known to be undecidable and Ackermann-complete respectively. Even for the class of deterministic OCA, which we show history-deterministic OCNs can be converted to, the inclusion problem is known to be undecidable.

\paragraph*{Organisation of the paper}
Section~\ref{sec:prelims} contains preliminaries where we introduce notation and define the concepts mentioned above rigorously. In Section~\ref{sec:Decide}, we show $\PSPACE$-completeness of checking if an input OCN is history-deterministic. 
In Section~\ref{sec:hdocnTodoca}, we show that the language expressed by history-deterministic one-counter  nets are contained in the language accepted by deterministic one-counter automata. Moreover, we discuss the complexity of checking language-inclusion, language-equivalence and universality of history-deterministic nets.
Finally, in the Section~\ref{sec:otherResults}, we analyse the changes in complexity when the counters are represented succinctly, or if zero tests are added. 
Due to space constraints, missing proofs can be found in the appendix.

%% file: prelim.tex
We use $\Sigma$ throughout to denote a finite set of alphabet, and $\Sigma^*$ to denote the set of all finite words consisting of letters from $\Sigma$. The empty word over $\Sigma$ shall be denoted by $\epsilon$. We  use $\Sigma_{\epsilon}$ to denote the set $\Sigma \cup \{\epsilon\}$. A language $\Lc$ over $\Sigma$ is a subset of $\Sigma^*$.

\paragraph{Labelled Transition System}
A \emph{labelled transition system} (LTS) is a tuple $\Sc$ consisting of $\Sc = (Q,\Sigma_{\epsilon}, \rightarrow, q_0,F)$. In this paper, we assume that
 $Q$ is a (countable) set of states, $q_0\in Q$ is the initial state,
$F\subseteq Q$ is the set of final states, $\Sigma$ is a finite alphabet, $\rightarrow \subseteq Q\times \Sigma_{\epsilon} \times Q$ is the set of transitions.
    
If a a transition $(q_1,a,q_2)$ belongs to $\rightarrow$, we instead represent it as $q_1\xrightarrow{a} q_2$ as well.
On a (finite) word $w$, a $\rho$ is said to be a (finite) \emph{run} of the labelled transition system $\Ac$ if it is an (finite) alternating sequence of states and letters of $\Sigma$:
$\rho = q_0\xrightarrow{a_0}q_1\xrightarrow{a_1}\dots q_{k-1}\xrightarrow{a_1} q_k$, 
where each $q_i\xrightarrow{a_i}q_{i+1}\in \rightarrow$ and $a_0\cdot a_1\dots a_k = w$ and $a_i\in\Sigma_\epsilon$.   
A run $\rho$ described above is accepting if the state $q_k\in F$.

An LTS that has no $\epsilon$-transitions is said to be a \emph{realtime LTS}. For an LTS $\Sc = (Q,\Sigma, \rightarrow, q_0,F)$ being realtime, we have $\rightarrow\subseteq Q\times \Sigma\times Q$. Unless mentioned otherwise, we mostly deal with realtime LTS for the sake of a simpler presentation.

An LTS $\Sc = (Q,\Sigma, \rightarrow, q_0,F)$ is \emph{deterministic} if $\rightarrow$ is a function from $Q\times \Sigma$ to $Q$, and not just a relation.

\paragraph*{Two player games}
Throughout the paper, we will be using two player games on countably sized arenas, between the players Adam and Eve, denoted by $\adam$ and $\eve$ respectively. The winning condition will be a reachability condition for one of the players, often $\adam$. By the work of Martin \cite{Mar75}, we know that such games are determined, that is they have a winner, which is either $\adam$ or $\eve$. Moreover, each of the players have a positional strategy, where their current strategy depends on their positions in the current arena. We shall say that two games are \emph{equivalent}, if they have the same winner. 

\paragraph*{One-Counter Automata}
A \emph{one-counter automaton} (OCA) $\Ac$ is given by a tuple $\Ac = (Q,\Sigma,\Delta,q_0,F)$, where $Q$ is a finite set of states, $q_0 \in Q$ is the initial state, 
 $F\subseteq Q$ is the set of final states,
 $\Sigma$ is a finite alphabet, and
finally, $\Delta$ is the set of transitions, given as a relation $\Delta \subseteq Q \times \{\zerotest,\notzerotest\} \times \Sigma \times \{-1,0,1\} \times Q $.
    
Here, the symbols $\zerotest$ and $\notzerotest$ are used to distinguish between transitions that can happen when the counter value is 0, and when the counter value is positive respectively. One can think of the counter as a `stack', where the stack has a distinguished bottom-of-the-stack symbol, which cannot be popped. The configurations in the automaton are given by pairs $(q,m)$,  where $q$ denotes the current state, and $m \in \mathbb{N}_0$ denotes the counter value. We use $\Cc(\Ac)$ to denote the set of configurations of $\Ac$. 

A one-counter automaton can be viewed as a succinct description of an infinite-state LTS over the set of configurations, such that the configurations are as defined below. For each configuration $(q,m)$, upon reading $a \in \Sigma_{\epsilon}$, 
\begin{itemize}
    \item if $m>0$, takes a transition of the form $(q,\notzerotest,a,d,q')$, where $d\in \{-1,0,1\}$ to $(q',m+d)$;
    \item if $m = 0$, takes a transition of the form $(q,\zerotest,a,d,q')$, where $d \in \{0,1\}$ to $(q',m+d)$.
\end{itemize} 
For two configurations $c,c' \in \Cc(A) = Q \times \mathbb{N}_{0}$, we use the notation $c \xrightarrow{a,d} c'$ to denote the fact that $c'$ can be reached from $c$ upon taking some transition $\delta \in \Delta$ upon reading $a$, with a change of counter value $d$. We shall also say that $c \xrightarrow{a,d} c'$ is a transition in $\Ac$, as $c \xrightarrow{a,d} c'$ is a transition in the infinite LTS of $\Ac$. We thus view $\Ac$ as both an automaton and a LTS, and switch between these two notions interchangeably. 
A run of $\Ac$ over a word $w$ is a finite sequence of alternating configurations and transitions : $\rho = c_0 \xrightarrow{a_0,d_0} c_1 \cdots c_n \xrightarrow{a_n,d_n} c_{n+1}$ such that $a_0a_1 \cdots a_n = w$, and $c_0 = (q_0,0)$. The run $\rho$ is an \emph{accepting run} if its last configuration $c_{n+1}=(q_{n+1},k_{n+1})$ is accepting, i.e. $q_{n+1} \in F$. We say a word $w$ is an \emph{accepting word} in $\Ac$ if it has an accepting run in $\Ac$. Finally, we define the language of $\Ac$, denoted by $\Lc(\Ac)$ to be the set of all accepting words in $\Ac$. 
We say that a one-counter automaton $\Ac$ is a \emph{deterministic one-counter automaton}, if $\Delta$ is a (partial) function from $ Q \times \{\zerotest,\notzerotest\} \times \Sigma$ to $ \{-1,0,1\} \times Q$.

\paragraph{One-counter nets} The model of
\emph{one-counter nets} (OCNs) can be interpreted as a restriction added to one-counter automaton that do not have the ability to test for zero. Alternatively, one can view this as a finite-state automaton that has access to a stack which can store only one symbol and no bottom-of-the-stack element. Any feasible run cannot pop an empty stack. More formally, a one-counter net $\Nc$ is a tuple $(Q, \Sigma, \Delta,q_0,F)$ where  $Q$ is the set of finite states, $\Sigma$ is a finite alphabet, $q_0\in Q$ is the initial state and $F\subseteq Q$ is the set of final or accepting states. The set $\Delta\subseteq Q\times \Sigma\times \{-1,0,1\}\times Q$ are the transitions in the net $\Nc$.

The configurations of an OCN are similar to that of an OCA. It consists of a pair $(q, n)\in Q\times \mathbb{N}$. We shall use the notation $\Cc(\Nc) = Q \times \mathbb{N}$ to denote the set of configurations of $\Nc$. From a configuration $(q,n)$, we reach a configuration  $(p,n + d)$ in one step, if there is a transition $\delta = (q,a,d,p)$, for some $a\in\Sigma$ and $d \in \{-1,0,+1\}$ and $n+d\geq 0$.
We can define a run on an OCN, an accepting run and an accepting word similar to an OCA. We shall say an OCN $\Nc$ is \emph{complete}, if for every configuration $c \in \Cc(\Nc)$ and every letter $a \in \Sigma$, there exists a transition $c \xrightarrow{a,d} c'$.

\begin{remark}
For the most of the paper we talk about one-counter nets (automata) with unary transitions, i.e. transitions that increment or decrement the counter by at most 1. However, they are as expressive as succinct models where the transitions are given in binary. This can be observed, for instance, by giving a construction similar to that of Valiant's for deterministic pushdown automata (Section 1.7, \cite{Val73}). 

\end{remark}

\paragraph*{History-Deterministic One-Counter Nets}
We define history-determinism in the setting of one-counter net. We say an OCN $\Nc$ is \emph{history-deterministic}, if the non-deterministic choices required to accept a word $w$ which is in $\Lc(\Nc)$ can be made on-the-fly. These choices depend only on the word read so far, and do not require the knowledge of the future of the word to construct an accepting run for a word in $\Lc(\Nc)$ (hence the term history-determinism). Formally, we say an OCN $\Nc$ is history-deterministic, if $\eve$ wins the letter game on $\Nc$ defined below.
\begin{definition}[Letter game for OCN]
Given an OCN $\Nc=(Q,\Sigma,\Delta,q_0,F)$, the letter game on $\Nc$ is defined between the players $\adam$ and $\eve$ as follows: the positions of the game are $\Cc(\Nc) \times \Sigma^*$, with the initial position $((q_0,0),\epsilon)$. At round $i$ of the play, where the position is $(c_i,w_i)$:
\begin{itemize}
    \item $\adam$ selects $a_{i} \in \Sigma$
    \item $\eve$ selects a transition $\delta$ which can be taken at the configuration $c_i$ on reading $a_i$, i.e. $c_i \xrightarrow{a_i,d_i}c_{i+1}$

\end{itemize}
If $\eve$ is unable to choose a transition (i.e. there is no $a_i$ transition at the configuration $c_i$ in the LTS generated by the net $\Nc$), and $w_{i+1} = w_i a_i$ is the prefix of an accepting word, $\eve$ loses immediately. The player $\adam$ wins immediately when the word $w_{i+1}$is accepting but the configuration $c_{i+1}$ is not at an accepting state, and the game terminates. The game continues from $(c_{i+1},w_{i+1})$ otherwise. The player $\eve$ wins any infinite play. 
\end{definition}
We say a strategy for $\eve$ in the letter game of $\Nc$ is a \emph{resolver} for $\Nc$, if it is a winning strategy for $\eve$ in the letter game.

Our characterization of history-deterministic one-counter nets by the above letter game is slightly different from the one presented in the work of Guha, Jecker, Lehtinen and Zimmermann~\cite{GJLZ21} for pushdown automata. In their work, they define history-determinism as having a consistent strategy based on the transitions taken so far. It is easy to argue that these two definitions are equivalent. 

The letter game can be formulated as a reachability game over countably many vertices, where the player $\adam$ is trying to reach a position of the form $(c,w) \in \Cc(\Nc) \times \Sigma^{*}$, where $c$ is at a rejecting state, while $w$ is accepting. Such games are determined, and this follows from Martin's Theorem~\cite{Mar75} showing that history-determinism formulated as a letter game is well-defined. 

Letter games have been used extensively to characterise history-determinism for other models as well, such as parity automata~\cite{HP06} and for various kinds of quantitative and timed automata on both finite and infinite words~\cite{BKLS20,BL21,HLT22}.  

To aid our understanding of history-determinism as well as the above definition, we provide an example of a game where $\eve$ wins the letter game on this automaton but the strategy is based on her counter configuration. 

\begin{example}\label{ex:HDOCN}
Consider the language 
$$\Lc = \{a^n\$ b^{n_1}\$ b^{n_2}\$\dots\$ b^{n_k}\$\mid \sum_{i=1}^kn_i\leq n \text{ or }  n_k = 2, \sum_{i=1}^{k-1}n_i = n-1\}.$$
which can be accepted by a history-deterministic OCN as shown in Figure~\ref{fig:HDOCN}. The initial state is indicated with an arrow pointing to it, and the final states are double-circled. Missing transitions are assumed to go to a rejecting sink state. In the corresponding letter game, $\adam$ plays the letter $a$ several times, say $n$-many times followed by a $\$$. The corresponding transitions so far are deterministic. Later, $\adam$ reads some series of $b$s and $\$$s, such that the word continues to be in the language. Note that the non-determinism occurs in only one state, which is marked with an $X$, upon reading the letter $b$. A winning strategy of $\eve$ which proves that this net is history-deterministic is the following: she takes the `down' transition if the counter value is strictly larger than $1$, but the `right' transition on $b$ otherwise. This non-determinism can't be determinised by removing transitions, because removing either of the `down' $b$-transition or the `right' $b$-transition changes the language accepted.
\end{example}

\begin{figure}
    \centering
    \makebox[\textwidth][c]{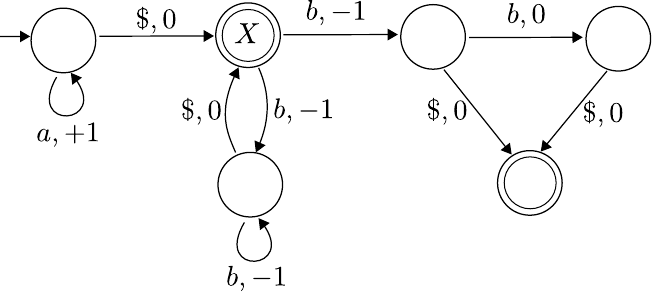}%
     \caption{A history-deterministic OCN accepting $\Lc$~\label{fig:HDOCN}}
\end{figure}

%% file: HDOCN.pdf_tex
\begingroup%
  \makeatletter%
  \providecommand\color[2][]{%
    \errmessage{(Inkscape) Color is used for the text in Inkscape, but the package 'color.sty' is not loaded}%
    \renewcommand\color[2][]{}%
  }%
  \providecommand\transparent[1]{%
    \errmessage{(Inkscape) Transparency is used (non-zero) for the text in Inkscape, but the package 'transparent.sty' is not loaded}%
    \renewcommand\transparent[1]{}%
  }%
  \providecommand\rotatebox[2]{#2}%
  \newcommand*\fsize{\dimexpr\f@size pt\relax}%
  \newcommand*\lineheight[1]{\fontsize{\fsize}{#1\fsize}\selectfont}%
  \ifx\svgwidth\undefined%
    \setlength{\unitlength}{187.56080169bp}%
    \ifx\svgscale\undefined%
      \relax%
    \else%
      \setlength{\unitlength}{\unitlength * \real{\svgscale}}%
    \fi%
  \else%
    \setlength{\unitlength}{\svgwidth}%
  \fi%
  \global\let\svgwidth\undefined%
  \global\let\svgscale\undefined%
  \makeatother%
  \begin{picture}(1,0.44687854)%
    \lineheight{1}%
    \setlength\tabcolsep{0pt}%
    \put(0,0){\includegraphics[width=\unitlength,page=1]{HDOCN.pdf}}%
  \end{picture}%
\endgroup%

%% file: HDviaToken.tex
The main result of this section is that  deciding history-determinism for a given OCN is decidable and is $\PSPACE$-complete as stated in the theorem below. 
\begin{theorem}\label{thm:PSPACE}
Given a one-counter net $\Nc$, checking if $\Nc$ is history-deterministic is $\PSPACE$-complete.  
\end{theorem}

The rest of this section is dedicated to the proof of the above statement. 

The proof of showing the upper bound proceeds by a series of polynomial reductions as below. 
$$\text{Deciding history-determinism}$$
$$\big\Updownarrow $$
$$\text{Deciding if }\eve \text{ wins letter game}$$ 
$$\big\Updownarrow $$
$$\text{Deciding if }\eve \text{ wins 1-token game}$$
$$\big\Downarrow$$
$$\text{Deciding if }\eve \text{ wins simulation game}$$

We shall define these games rigorously and prove these reductions in Subsection~\ref{subsec:TokenG}. Finally, since the simulation problem for one-counter nets is in $\PSPACE$~\cite{HLMT16}, this gives us the upper bound.

For the lower bound, we reduce from the problem of emptiness checking for alternating finite-state automata over a unary alphabet to deciding if $\eve$ wins the letter game. 

\subsection{Token games}\label{subsec:TokenG}
Deciding history-determinism efficiently for finite-state parity automata over infinite words has been a major area of study over the recent years. Bagnol and Kupergerg~\cite{BK18}, gave a polynomial time procedure for deciding history-determinism when the finite automata accepts with a B\"uchi condition. Their underlying technique is a two-player game, called $G_2$ or 2-token games, which they proved to be equivalent to the letter game when the automaton is B\"uchi. Boker, Kuperberg, Lehtinen and Skrzypczak~\cite{BKLS20} extended this to show that the game $G_2$ is equivalent to the letter game when the automaton is co-B\"uchi as well. Deciding the winner in $G_2$ for an automaton of a fixed parity index takes polynomial time~\cite{BKLS20}, and hence deciding history-determinism for the cases of when the parity automata accepts words based on B\"uchi or co-B\"uchi condition is polynomial. It is famously conjectured that winning $G_2$ is equivalent to the letter game for higher parity indices as well, and this is known as the $G_2$ conjecture~\cite{BKLS20}. Token games have also been instrumental in deciding history-determinism for quantitative automata, in the works of Boker and Lehtinen~\cite{BL22}. In their paper, they show that for finite words on a finite-state boolean automaton, history-determinism is characterised by $G_1$. This was later extended to labelled transition systems with countably many states, in the works of Henzinger, Lehtinen and Totzke~\cite{HLT22}. Thus, the $1$-token games also characterise history-determinism over OCNs. We include a proof nonetheless, for the sake of completeness.

In a play of the letter game, $\adam$ picks the letters while $\eve$ picks the transitions, and the winning condition for $\eve$ is to produce an accepting run for any word that is in the language. Token games work similarly, but they impose more restrictions on $\adam$. This is done by asking him to also display a valid run during the game with the help of some number of tokens. Here, we concentrate on the 1-token game $G_1$. The player $\adam$ wins the game $G_1$ if and only if he produces an accepting run, whilst $\eve$ produces a rejecting run. We make this more formal in the definition below.

\begin{definition}[One token game $G_1$]
Let $\Nc=(Q,\Sigma,\Delta,q_0,F)$ be a one-counter net. The positions of the game $G_1$ on $\Nc$ are a pair of configurations, $\Cc(\Nc) \times \Cc(\Nc)$, where the first configuration is $\eve$'s token, and the second is $\adam$'s token. The game starts with the initial position $(c^{\eve}_0,c^{\adam}_0)=((q_0,0),(q_0,0))$. At the $i^{th}$ iteration of the play, where the position is $(c^{\eve}_i,c^{\adam}_i)$:
\begin{enumerate}
    \item $\adam$ selects $a \in \Sigma$
    \item $\eve$ selects a transition for her token, $ c^{\eve}_i \xrightarrow{a,d} c^{\eve}_{i+1} $ 
    \item $\adam$ selects a transition for his token, $c^{\adam}_i \xrightarrow{a,d'} c^{\adam}_{i+1} $ 
\end{enumerate}
If $\eve$ is unable to choose a transition for her token whereas $\adam$ can choose a transition and extend the run on his token to an accepting run, then the game terminates and $\eve$ loses the game. However, irrespective of $\eve$'s ability to extend her run, if $\adam$ is unable to choose a transition for his token, then the game again terminates but $\adam$ loses the game.

If both the players can extend their runs by picking a transition then and if  $\adam$'s state in $c^{\adam}_{i+1}$ is accepting, but $\eve$'s state in $c^{\eve}_{i+1}$ is rejecting then again the game terminates and $\eve$ loses the game. Else, the game goes to $(c_{i+1},c'_{i+1})$ for another round of the play. We add that $\eve$ wins any infinite play.
\end{definition}
We show in the following lemma that  $\adam$, even with limited power, in one-token games can capture letter games. Letter games can be seen as a version of token games where $\adam$ plays with infinitely many tokens.

\begin{lemma}\label{lemma:G1HD}
For a OCN $\Nc$, if $\eve$ wins the game $G_1$ on $\Nc$, then $\eve$ has a winning strategy in the letter game as well.
\end{lemma}

To prove the above lemma, we need to understand better the structure of the resolvers for OCNs. 
Consider the definition given below of \emph{residual transitions}. Intuitively, these are transitions such that if there was an accepting word from a configuration with the first letter as $a$, then upon taking a residual transition on $a$, there is still an extension of the run on the word from the new configuration that is accepting.
More formally, we say that a transition $(q,k) \xrightarrow{a,d} (q',k')$ is \emph{residual} if $\Lc(q',k') = a^{-1}\Lc(q,k)$, where $\Lc(q,k)$ (and $\Lc(q',k')$) is the set of words that are accepted in $\Nc$ when the initial configuration is $(q,k)$ ($(q',k')$), instead of $(q_0,0)$. 
The proposition below shows any winning strategy of $\eve$ can be characterised by these residual transitions.

\begin{proposition}\label{prop:residualresolver}
For an OCN $\Nc$, an $\eve$ strategy $\sigma$ in the letter game is winning for $\eve$ if and only if  $\sigma$ takes only residual transitions. 
\end{proposition}

Note that in the letter game, each player winning the game has a positional winning strategy, as it is a reachability game. Suppose that $\eve$ wins the letter game, then $\eve$ has a winning strategy which can be given by a (partial) function $\sigma :(Q \times \mathbb{N}) \times \Sigma^{*} \times \Sigma \rightarrow \Delta^{*}$. Using  Proposition~\ref{prop:residualresolver}, we can show that $\eve$'s strategy only depends on the configuration, and is independent of the word read so far. 
\begin{proposition}\label{prop:positional}
If $\eve$ wins the letter game, then $\eve$ has a winning strategy $\sigma$ that only depends on the current configuration of the play, i.e $\sigma$ is a partial function $\sigma : (Q \times \mathbb{N}) \times \Sigma \rightarrow \Delta^{*}$    
\end{proposition}

%% file: tokengames.tex
Having shown that $G_1$ is equivalent to the letter game, we show that deciding the winner in the game $G_1$ is decidable in $\PSPACE$ (when the transitions are unary). This implies deciding history-determinism is also decidable, and in $\PSPACE$. We do so by reducing $G_1$ to the simulation problem between two one-counter nets, which is known to be $\PSPACE$-complete (cf. Theorem 7,~\cite{HLMT16}).   
Given two one-counter nets $\Nc$ and $\Nc'$ at configurations $(q,n)$ and $(q',n')$, intuitively, we say $\Nc'$ simulates $\Nc$ (or $\Nc$ is simulated by $\Nc'$) from their corresponding configurations if for any sequence of transitions from $(q,n)$, there is also a sequence of transitions from $(q',n')$ which is built `on-the-fly'. This alternation between existential and universal quantifiers in the above statement  renders this definition perfect to be captured by the following simulation game between two players $\adam$ and $\eve$. 

\begin{definition}[Simulation Game]\label{def:simulationgame}
Given two OCNs  $\Nc = (Q, \Sigma, \Delta,q_I, F)$ and $\Nc' = (Q',q'_0,\Sigma,\Delta',q'_I, F')$ and two configurations $c = (p,k)$ and $c' = (p',k')$ in $\Cc(\Nc)$ and $\Cc(\Nc')$ respectively where $k,k' \in \mathbb{N}$. The \emph{simulation game} between the OCNs $\Nc$ and $\Nc'$ at a position $(c,c')$, denoted by $\simg{(\Nc,c)}{(\Nc',c')}$, is a two player game between $\adam$ and $\eve$, with positions in $\Cc(\Nc) \times \Cc(\Nc')$ where the initial position is $(c_0,c'_0) = (c,c')$. At round $i$ of the play, where the position is $(c_i,c'_i)$:
\begin{itemize}
    \item $\adam$ selects a letter $a \in \Sigma$, and a transition $c_i \xrightarrow{a,d} c_{i+1}$ in $\Nc$
    \item $\eve$ selects an $a$-transition $c'_i \xrightarrow{a,d'} c'_{i+1}$ in $\Nc'$
\end{itemize}
If $\adam$ is unable to choose a transition, then $\adam$ loses the game immediately. If $\eve$ is unable to choose a transition but $\adam$ can select a transition and extend the run in $\Nc$ to an accepting run, then $\eve$ loses the game.

Otherwise, if  $\adam$'s state in $c_{i+1}$ is accepting but $\eve$'s state in $c_{i+1}'$ is rejecting, then $\eve$ loses the game, and the game terminates. Else, the game goes to $(c_{i+1},c'_{i+1})$ for another round of the play. The player $\eve$ wins any infinite play.
\end{definition}

If $\eve$ wins the above game, we say $(\Nc',(p',k'))$ simulates $(\Nc,(p,k))$, and we denote it by $(\Nc,(p,k))\simulates (\Nc',(p',k'))$. Furthermore, we say $\Nc'$ simulates $\Nc$ or $\Nc \simulates \Nc'$ if $(\Nc,(q_I,0))\simulates (\Nc',(q'_I,0))$.

As the simulation game is a reachability game over a countably sized arena, it is determined, and the winning player has a positional strategy. Thus, if $\eve$ wins the above simulation game $\simg{\Nc,(p,k))}{\Nc',(p',k')}$, then $\eve$ has a positional winning strategy $\sigma_{\eve} : \Cc(\Nc) \times \Cc(\Nc') \times \Sigma \rightarrow \Delta'$. 

\begin{remark}\label{remark:simulationFinalstate}
In the literature over one-counter nets~\cite{Srb06,HLMT16,JOS18}, the winning condition for the players on the simulation game is expressed differently, via the inability of the players to choose transitions, rather than accepting states. The player $\adam$ ($\eve$) loses the game if $\adam$ ($\eve$) is unable to choose a transition. It can however, be shown that the two versions of the simulation games are log-space reducible to each other. We show this equivalence in Appendix~\ref{appendix:Sim}.
\end{remark}

Note the similarities (and differences) in $G_1$ and the simulation game. In both, the winning condition for $\adam$ would like $\adam$'s run to be accepting, while $\eve$'s to be rejecting. In $G_1$ however, $\eve$ is picking the transition first, while in the simulation game, $\adam$ is picking the transition first. 

With some modifications to the structure of the underlying nets in $G_1$, we can ensure that the simulation game between the modified net and the original net captures $G_1$. The intuition is that, in the simulation game, the net which is simulated is modified such that $\adam$ is forced to delay choosing his transition.
This is formalized in the proof of the following lemma, and explained with a diagram in Figure~\ref{fig:G1Sim}.  

\begin{figure}
    \centering
    \makebox[\textwidth][c]{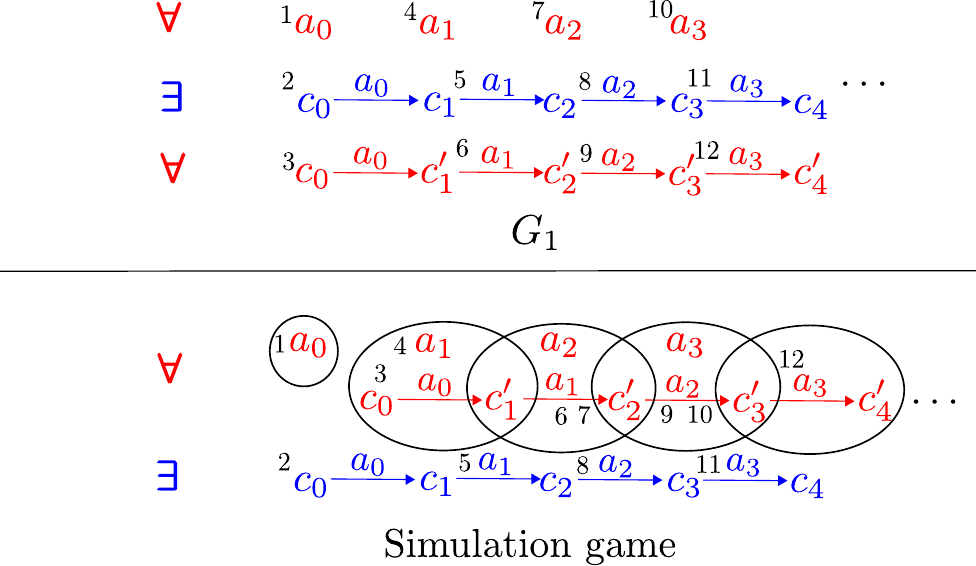}%
    \caption{An illustration of a play of $G_1$, seen as a play of the simulation game\label{fig:G1Sim}}
\end{figure}

\begin{lemma}\label{lemma:G1Simulation}
Given a one-counter net $\Nc$, there are one-counter nets $\Mc$ and $\Mc'$, which have size at most polynomial in size of $\Nc$ such that $\eve$ wins $G_1$ on $\Nc$ if and only if $\eve$ wins $\Mc \simulates \Mc'$. 
\end{lemma}
\begin{proof}(Sketch)
 Figure~\ref{fig:G1Sim} captures the intuition behind the proof. Here, note that we have different linearisations of the play, but the alternation between $\eve$ and $\adam$ required is captured by the simulation game by making $\adam$ choose his configuration and transition at the same time. 
For each run in $\Nc$, we have a run in $\Mc$ that lags behind one transition, and it does so by remembering which letter it should move on next. We provide a construction such that $\Mc'$ is linear in the size of $\Nc$ and $\Mc$ has size approximately $\Nc\times|\Sigma|$, where $|\Sigma|$ is the size of the alphabet. This factor of $|\Sigma|$ arises due to remembering in the state space, the previous letter read, to create a lag for $\adam$'s decisions. We then show that the player $\eve$ wins $G_1$ on $\Nc$ if and only if $\Mc  \simulates \Mc'$. 
\end{proof}

Finally, we see that the following theorem from the work of Hofman, Lasota, Mayr and Totzke~\cite{HLMT16} shows that the winner of a simulation game can be solved in $\PSPACE$. We recall their results to fit our notation below.

\begin{theorem}\label{theorem:outsource}
Given two one-counter nets $\Nc$ and $\Nc'$, with configurations $(p,k)$ and $(p',k')$ in $\Cc(\Nc)$ and $\Cc(\Nc')$ respectively, with $k$ and $k'$ represented in binary, deciding whether $(\Nc',(p',k'))$ simulates $(\Nc,(p,k))$ is in $\PSPACE$. Moreover, the set of $(k,k')$ for which $(\Nc,(p,k))\simulates (\Nc',(p',k'))$ is semilinear, and can be computed in $\EXPSPACE$.  
\end{theorem}
\begin{proof}
See \cite{HLMT16}, cf. Theorem 7
\end{proof}

\begin{lemma}\label{lemma:UB}
Given a one-counter net $\Nc$, we can decide in $\PSPACE$ if $\Nc$ is history-deterministic.
\end{lemma}
The above lemma is a corollary of Lemma~\ref{lemma:G1HD}, Lemma~\ref{lemma:G1Simulation} and Theorem~\ref{theorem:outsource}.
\input{lowerbound.tex}

%% file: G1Sim.pdf_tex
\begingroup%
  \makeatletter%
  \providecommand\color[2][]{%
    \errmessage{(Inkscape) Color is used for the text in Inkscape, but the package 'color.sty' is not loaded}%
    \renewcommand\color[2][]{}%
  }%
  \providecommand\transparent[1]{%
    \errmessage{(Inkscape) Transparency is used (non-zero) for the text in Inkscape, but the package 'transparent.sty' is not loaded}%
    \renewcommand\transparent[1]{}%
  }%
  \providecommand\rotatebox[2]{#2}%
  \newcommand*\fsize{\dimexpr\f@size pt\relax}%
  \newcommand*\lineheight[1]{\fontsize{\fsize}{#1\fsize}\selectfont}%
  \ifx\svgwidth\undefined%
    \setlength{\unitlength}{281.0240505bp}%
    \ifx\svgscale\undefined%
      \relax%
    \else%
      \setlength{\unitlength}{\unitlength * \real{\svgscale}}%
    \fi%
  \else%
    \setlength{\unitlength}{\svgwidth}%
  \fi%
  \global\let\svgwidth\undefined%
  \global\let\svgscale\undefined%
  \makeatother%
  \begin{picture}(1,0.57911365)%
    \lineheight{1}%
    \setlength\tabcolsep{0pt}%
    \put(0,0){\includegraphics[width=\unitlength,page=1]{G1Sim.pdf}}%
  \end{picture}%
\endgroup%

%% file: lowerbound.tex
\subsection{Lower Bounds}

Although solving the simulation game turns out to be $\PSPACE$-complete itself from the work of Srba~\cite{Srb06}, this lower bound result does not work for our reduction to simulation games. The reduction we give from $G_1$ to simulation games produces only a restricted class of simulation games which solve $G_1$. 

Nevertheless, we show that deciding history-determinism is still $\PSPACE$-hard, showing that even this restriction of the simulation problem is enough to induce $\PSPACE$-hardness.

\begin{lemma}\label{lemma:LB}
Given a one-counter net $\Nc$, it is $\PSPACE$-hard to decide if $\Nc$ is history-deterministic.
\end{lemma}
\begin{proof}[Sketch]
The proof goes by reducing from the problem of checking non-emptiness of an alternating finite-state automaton over a unary alphabet. This variation of the problem was proven to be $\PSPACE$ complete by Holzer~\cite{Hol95}, with its proof simplified by Jan\v{c}ar and Sawa~\cite{JS07}.
The intuition behind the construction is to recreate a run of the alternating automaton using the constructed net. In the letter game, a fair play of $\adam$ corresponds to a branch of a run-tree in the automaton, with $\eve$ resolving universal transitions and $\adam$ resolving existential ones. The player $\adam$ can ensure that he wins the letter game if and only if the alternating automaton has some word that he can demonstrate is in the language. If $\adam$ plays unfairly, then there are gadgets to ensure that $\eve$ automatically wins.
\end{proof}
We conclude this section by saying that
Lemma~\ref{lemma:UB} and Lemma~\ref{lemma:LB} together give the proof of Theorem~\ref{thm:PSPACE}

%% file: rhocn.tex
We dedicate this section to tackling different questions about languages accepted by history-deterministic one-counter nets and decision problems on such languages.

\subsection{Languages Accepted by History-Deterministic OCNs}
While in history-deterministic models we are able to resolve the non-determinism on-the-fly, it is not well-understood how these resolvers might look like in general. In fact, Guha, Jecker, Lehtinen and Zimmermann showed that there are history-deterministic pushdown automata whose resolvers cannot be given by a pushdown automata~\cite{GJLZ21}, and whether such a resolver can be computed is an open problem.

In this sub-section, our goal is to understand better the languages of history-deterministic OCNs. As a first-step towards this goal, we already have some intuition from the previous section on the eventually periodic nature of the transitions that are residual (as a corollary of Lemma~\ref{lemma:G1Simulation} and Theorem~\ref{theorem:outsource}). Here, we solidify this intuition by defining what it means to have \emph{\ssp} for a resolver and to then show that all nets have this property. For the case of history-deterministic nets, using this semi-linearity of the resolvers, we show the existence of a language-equivalent deterministic OCA.

We first show a sufficient characterisation which we call the \emph{\ssp}, for if a given history-deterministic one-counter net can be determinised. 

We say a transition $\delta = (p,a,d,p')$ in an one-counter net $\Nc$ is a $\gut$ transition at $(p,k)$, if $((p,k),(p,k))$ is in the winning region of $G_1$, and the transition $\delta = (p,k) \xrightarrow{a,d} (p',k+d)$ is a winning move for $\eve$ when $\adam$ chooses the letter $a$. We also write this sometimes as $(p,k) \xrightarrow{a,d} (p',k+d)$ is a $\gut$ transition in $\Nc$. The following lemma can be seen as a weakening of Proposition~\ref{prop:residualresolver} :
\begin{lemma}\label{lemma:guttransitions}
Let $\Nc = (Q, \Sigma,\Delta,q_0,F)$ be a history-deterministic one-counter net. An $\eve$ strategy $\sigma$ in the letter game is winning for $\eve$ if and only if
    the strategy $\sigma$ only takes $\gut$ transitions $ \delta = (p,k) \xrightarrow{a,d} (p',k')$.
\end{lemma}
\begin{proof} 
Note that any strategy of $\eve$ that is winning in the letter game takes only $\gut$ transitions, as $G_1$ is a weaker game for $\eve$ than the letter game. The other direction follows by observing that any $\gut$ transition is also residual. If $(p,k) \xrightarrow{a,d} (p',k')$ is $\gut$, then for any word $aw \in L(p,k)$, we must have $w\in L(p',k')$. If not, then $\adam$ can win $G_1$ by constructing an accepting run on $aw$ from $(p,k)$ which contradicts the definition of $\gut$ transitions.  Hence the proof follows from Proposition~\ref{prop:residualresolver}. 
\end{proof}

\begin{definition}\label{def:ssp}
Given a one-counter net $\Nc$, we say $\Nc$ satisfies \ssp ~if for each transition $\delta =(q,a,d,q')$, the set of $k \in \mathbb{N} $ such that $\delta$ is a $\gut$ transition at $(q,k)$ is semilinear. That is for each transition $\delta = (q,a,d,q') \in \Delta$, we have that the set $$\Sc_{\delta} = \{k : (q,k) \xrightarrow{a,d} (q',k')\text{ is a \gut \ transition at} \ (q,k) \}$$ is semilinear.
\end{definition}

Consider the following example which solidifies this intuition: 
\begin{figure}
    \centering
    \makebox[\textwidth][c]{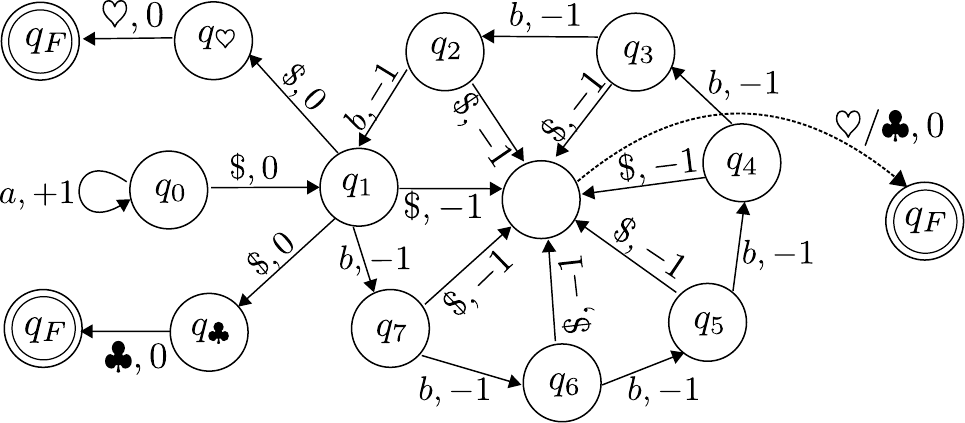}%
    \caption{The one-counter net $\Nc_7$ from Example~\ref{example:N7}~\label{fig:N7}}
\end{figure}
\begin{example}\label{example:N7}
Consider the net $\Nc_{7}$, as shown in Figure~\ref{fig:N7}, where all states labelled $q_F$ are accepting. This automaton is not  history-deterministic. However, if the counter value at $q_1$ is not a multiple of $7$, then $\eve$ can resolve the non-determinism from $q_1$.
Observe that the automaton accepts words of the form $a^n\$b^k\$\cdot(\heartsuit,\clubsuit)$ such that $k\leq n$. 
Consider the following play of $\adam$ in the letter game from $q_0$: For $7n$ steps he reads $a$, after which he reads a $\$$. So far, all transitions are deterministic. 
After that, assume he again reads, $7n$ many times, the letter $b$. This ensures that the transition ends at state $q_1$ with counter value 0. If he reads $\$$ here, this is the only position where $\eve$ has a choice. Note that she has to choose between $q_{\heartsuit}$ and $q_{\clubsuit}$ and since both the suffix  $\heartsuit$ and $\clubsuit$ are accepting, she loses no matter what she picks. However, if $\adam$ had read a number of `$b$'s was not a multiple of $7$, the play of an accepting word would end at $q_{\$}$ which is accepting. This serves to show two things: firstly, a non-example of history-determinism, and secondly, how the counter values affect the decisions of the player, which is in this case, $\adam$.
\end{example}

\begin{lemma}\label{lemma:DetBySSP}
If a history-deterministic one-counter net $\Nc = (Q, \Sigma,q_0,\Delta,F)$ satisfies the \ssp, then there is a language-equivalent deterministic OCA $\Dc$.
\end{lemma}
\begin{proof}[Sketch]
We assume the history-deterministic one-counter net $\Nc$ is such that it satisfies \ssp. 
We first construct a non-deterministic one-counter automata $\Bc$, which can be determinised easily by removing a minimal set of transitions to get rid of non-determinism while still preserving the language. The non-deterministic one-counter automata $\Bc$ would essentially be designed so that the transitions in $\Bc$ correspond to the $\gut$ transitions in $\Nc$, from any configuration. The eventual periodicity of the sets $S_{\delta}$ allows us to express this as a one-counter automaton, rather than as a labelled transition system with countably many states.

Intuitively, the automaton $\Bc$ is constructed such that the state space of the automaton stores in its memory the period and the initial block of the semi-linear sets. The idea is that this automaton's runs would be in bijection with those runs that take only $\gut$ transitions in the OCN $\Nc$. We know that such a run exists in $\Nc$ by Lemma~\ref{lemma:guttransitions}, as $\Nc$ is history-deterministic. However, the counter values in $\Bc$ are `scaled down' to only remember how many periods have passed, while counter value 0 indicates that the counter value in the original run would have been at most $I$. The exact value of the counter value in a run of $\Nc$ can be inferred as a function of the state space.    
\end{proof}
Having shown that if a history-deterministic one-counter net satisfies \ssp, then we have an equivalent DOCA, we proceed to show that every one-counter net satisfies \ssp.
\begin{lemma}\label{lemma:netSSP}
Every one-counter net $\Nc$ satisfies \ssp. 
\end{lemma}
The proof of the above lemma is similar to the proof of Lemma~\ref{lemma:G1Simulation}. As an easy corollary of the above two lemmas, we get the following theorem.

\begin{theorem}\label{theorem:OCN-DOCA}
Every history-deterministic OCN can be determinised to produce an equivalent deterministic OCA. 
\end{theorem}

An easy analysis of our proof combined with the results on the representation of simulation preorder (Lemma~28,~\cite{HLMT16}) shows a doubly exponential upper bound on the size of the equivalent deterministic OCA constructed from the proof of the theorem above. However, we conjecture that there exists an (at most) exponentially sized language-equivalent determistic OCA for every history-deterministic OCN.

\begin{remark}
On the topic of expressivity of history-determinism, we conclude this subsection with  a remark that history-deterministic OCNs are strictly less expressive than non-deterministic OCNs. This can be demonstrated with the following language
$$\Lc = \{a^i\$b^j\$b^k\mid j\leq i\text{ or }k\leq i\}.$$
It is routine to verify that such a language is not accepted by any history-deterministic OCN, but this language can be accepted by a non-deterministic OCN. Note that history-determinism itself is not the limiting factor in accepting this language, as this language is accepted by a history-deterministic pushdown automaton~\cite{GJLZ21}. 

\end{remark}

\subsection{Complexity of comparing languages of history-deterministic OCNs}

\input{HDNetproperties.tex}

%% file: N7.pdf_tex
\begingroup%
  \makeatletter%
  \providecommand\color[2][]{%
    \errmessage{(Inkscape) Color is used for the text in Inkscape, but the package 'color.sty' is not loaded}%
    \renewcommand\color[2][]{}%
  }%
  \providecommand\transparent[1]{%
    \errmessage{(Inkscape) Transparency is used (non-zero) for the text in Inkscape, but the package 'transparent.sty' is not loaded}%
    \renewcommand\transparent[1]{}%
  }%
  \providecommand\rotatebox[2]{#2}%
  \newcommand*\fsize{\dimexpr\f@size pt\relax}%
  \newcommand*\lineheight[1]{\fontsize{\fsize}{#1\fsize}\selectfont}%
  \ifx\svgwidth\undefined%
    \setlength{\unitlength}{277.73820318bp}%
    \ifx\svgscale\undefined%
      \relax%
    \else%
      \setlength{\unitlength}{\unitlength * \real{\svgscale}}%
    \fi%
  \else%
    \setlength{\unitlength}{\svgwidth}%
  \fi%
  \global\let\svgwidth\undefined%
  \global\let\svgscale\undefined%
  \makeatother%
  \begin{picture}(1,0.43808482)%
    \lineheight{1}%
    \setlength\tabcolsep{0pt}%
    \put(0,0){\includegraphics[width=\unitlength,page=1]{N7.pdf}}%
  \end{picture}%
\endgroup%

%% file: HDNetproperties.tex
The complexity of comparisons between languages of non-deterministic OCNs are undecidable~\cite{HMT13}, and even the restricted question of universality, is Ackermann-complete~\cite{HT14}. Whereas for deterministic one-counter automata, although equivalence and therefore universality is in $\NL$~\cite{BG11,BGJ13}, inclusion is undecidable~\cite{Val73}. In this section, we show that for history-deterministic nets, these problems are no longer undecidable and have a significantly lower complexity when compared to non-deterministic nets. 

Note that although we have a procedure to determinise our automaton earlier in this section, this procedure does not help us answer these questions. This is because our determinisation procedure results in a deterministic OCA rather than an deterministic OCN. For deterministic OCNs, all these problems are known to be $\NL$-complete~\cite{HT14}, but for deterministic OCA, the problem of inclusion is undecidable~\cite{Val73}. Even though equality and universality for a deterministic OCA is $\NL$ complete, the resulting deterministic OCA we get from determinisation of history-deterministic OCNs could be much larger than our input net, leading to much larger complexity.

Nevertheless, we show that checking language inclusion and hence checking language equivalence between two history-deterministic one-counter nets is in $\PSPACE$.
This is done by showing a reduction to the problem of deciding history-determinism. Recall that as this problem is in $\PSPACE$ from Lemma~\ref{lemma:UB} and Theorem~\ref{thm:PSPACE}, we are able to show membership in $\PSPACE$ for language equivalence and inclusion between two history-deterministic one-counter nets. Moreover, using results of Kucera~\cite{Kuc00}, we get decidability in $\P$ for language universality.

\begin{lemma}\label{lemma:HDinclusion}
Deciding language inclusion and language equivalence between two history-deterministic one-counter nets is in $\PSPACE$. 
\end{lemma}
We can show that the problem of checking language inclusion between two history-deterministic OCNs reduces to checking if a larger OCN (linear in the sum of the size of the two OCNs) is history-deterministic.
Since language equivalence is essentially checking language inclusion both ways, we have the above results. 

\begin{lemma}\label{lemma:universality}
Deciding language universality for a given history-deterministic one-counter net is in $\P$.
\end{lemma}
The problem of universality reduces to checking if the input net $\Mc$ simulates a finiste state automata. This problem was shown to be $\P$ by  Kucera~(Lemma 2, \cite{Kuc00}), showing that universality is in $\P$.

We therefore have the following theorem.
\begin{theorem}
For nets $\Hc$ and $\Hc'$ that are history-deterministic, the problem of checking if $\Lc(\Hc)\subseteq \Lc(\Hc')$ as well as checking if $\Lc(\Hc) = \Lc(\Hc')$  can be done in $\PSPACE$. If $\Hc$ is instead a deterministic finite-state automaton, this problem can be solved in $\P$.
\end{theorem}
We summarise known results and complexity of relevant results for comparison with other automata models in Table~\ref{table:table1}.
\begin{table}[ht]
    \centering
    \begin{tabular}{ | m{1.7cm} || m{2.8cm}| m{2.8cm} | m{2.8cm} | } 
  \hline
   & $\Lc\subseteq \Lc'$ & $\Lc= \Lc'$ & $\Lc = \Sigma^*$ \\ 
  \hline 
  \hline
  DOCN & $\NL$-complete~\cite{HT14}& $\NL$-complete~\cite{HT14} & $\NL$-complete~\cite{HT14}\\ 
  \hline
  HOCN & In $\PSPACE$ & In $\PSPACE$ & In $\P$\\ 
  \hline
  OCN & Undecidable~\cite{Val73}& Undecidable~\cite{HMT13} & Ackermann-complete~\cite{HT14}\\
  \hline
  DOCA & Undecidable~\cite{Val73}& $
  \NL$-complete~\cite{BG11}& $\NL$-complete~\cite{BG11}\\
  \hline
\end{tabular}
    \caption{Complexities for the problems of deciding language inclusion, equivalence and universality over deterministic OCN, history-deterministic OCN, non-deterministic OCN and deterministic OCA.~\label{table:table1}}
\end{table}

%% file: miscResults.tex
We revisit the question of deciding history-determinism in this section for one-counter nets and its variants. In the first subsection, we tackle the question of how the complexity changes if the encoding of these nets are given in binary. 
We show that as expected, this increases the complexity of the problem from $\PSPACE$-complete to $\EXPSPACE$-complete. 
We then answer affirmatively to the question of whether adding zero-tests add too much power to one-counter nets by showing that the problem of deciding history-determinism becomes undecidable. 
\subsection{Succinct Encoding of Counters} 
\input{succinct.tex}
\subsection{Deciding History-Determinsm for OCA}
\input{undecidability.tex}

%% file: succinct.tex
If the input nets are encoded succinctly, we show that the problem of deciding history-determinism becomes $\EXPSPACE$-complete. By a succinct representation here, we mean that whenever we allow for an increment and a decrement in our net, these values are encoded as in binary in the input. In this representation, we wish to see if deciding history-determinism for such nets is harder. Unsuprisingly, we can show that this problem takes $\EXPSPACE$ when the nets are encoded in binary, which we remark in the following proposition.
\begin{proposition}\label{lemma:succUB}
Given a net $\Nc$ is such that transitions allow for binary encoding of the value, then deciding if $\Nc$ is history-deterministic is in $\EXPSPACE$.
\end{proposition}
This result follows from the previous proof of $\PSPACE$ upper bound from Lemma~\ref{lemma:UB} of deciding history-determinism for one-counter nets, where counter values are in unary.
Any net with binary encoding can be converted with only an exponential blow-up into another language equivalent net with unary encoding, preserving history-determinism.
This naturally gives us an $\EXPSPACE$ upper bound.

However, much more work is needed to show a matching lower bound, which we do by giving a reduction from reachability games on succinct one-counter nets (SOCN). Intuitively, these games are played on the configuration graphs of a one-letter OCN, where the states of the OCN are partitioned among two players, which we denote by $\land$ and $\lor$. The goal of the $\lor$ is to be able to take the play to a designated winning state with value $0$.
This problem was shown to be $\EXPSPACE$-complete by  Hunter~\cite{Hun15} and later, several of its variants were also shown to have the same complexity~\cite{JOS18}. This therefore gives us our $\EXPSPACE$-completeness for deciding history-determinism of one-counter nets. 
The name succinct comes from the encoding of the net in the input.

\begin{lemma}~\label{lemma:succinctHard}
Given an OCN $\Nc$, where the numbers in the transitions are represented in binary, deciding if $\Nc$ is history-deterministic is $\EXPSPACE$-hard.
\end{lemma}
\begin{proof}[Sketch]
Given an instance of a SOCN-reachability game on $\Nc$, We construct an OCN $\Mc$ such that $\lor$ wins in the SOCN-reachability game on $\Nc$ if and only if $\adam$ wins in the letter game on $\Mc$.

The high-level idea of the construction is that we construct an automaton $\Mc$, such that in a play of the letter game on $\Mc$,  the players $\adam$ and $\eve$ create a transcript of a run of the automaton $\Nc$. This is done easily by $\adam$ picking the letters at $\lor$ states, where he can pick a different letter, each corresponding to a different transition. Since in the letter game, we have in $\eve$ to resolve the non-determinism, we do that to allow for $\eve$ to resolve the choices of the $\land$ player. 

However, we need to ensure a few important aspects while constructing $\Mc$. 
Firstly, any sequence of letters chosen by $\adam$ in $\Mc$'s letter game so far must correspond to a run in $\Nc$ and secondly, the interplay between $\eve$'s and $\adam$'s choices in the letter game of $\Mc$ must correspond to the choices of the player $\land$ and $\lor$ respectively in the SOCN-reachability game of $\Nc$. These are the main challenges while constructing such an OCN $\Nc$ and they are resolved by the use of a few gadgets that we describe in detail in the appendix. 
\end{proof}
We conclude this subsection with the following theorem.
\begin{theorem}~\label{thm:succinctcomp}
Given an OCN $\Nc$ where the numbers in the transitions are represented in binary, deciding if $\Nc$ is history-deterministic is $\EXPSPACE$-complete.
\end{theorem}

%% file: undecidability.tex
We show that, given a one-counter automaton $\Ac$, deciding if $\Ac$ is history-deterministic is undecidable. It was shown by Guha, Jecker, Lehtinen and Zimmermann~\cite{GJLZ21} that deciding if a non-deterministic pushdown automaton is history-deterministic is undecidable. This extends their result to OCAs. The reduction follows from the undecidability of language inclusion for deterministic one-counter automata (DOCA)~\cite{Val73}.

\begin{theorem}\label{thm:UndecOCA}
Given an OCA $\Ac$, deciding if $\Ac$ is history-deterministic is undecidable.
\end{theorem}
\begin{proof}[Sketch]
Consider the following problem : 
\begin{quote} 
\emph{DOCA Inclusion:} Given two DOCAs $\Ac$ and $\Bc$, is $\Lc(\Ac) \subseteq \Lc(\Bc)$?
\end{quote} 
The above problem was shown to be undecidable in Section 5.1 of Valiant's thesis~\cite{Val73}. We show that the problem of deciding if a given one-counter automaton is history-deterministic is also undecidable, by the means of a reduction. 
\end{proof}
This shows that zero-tests already add too much power for the problem of deciding history-determinism.

%% file: conclusion.tex
We showed several decision problems related to history-determinism  to be decidable over OCNs. This is unlike other classes of infinite-state systems that subsume them, where either a subset or all of the problems are undecidable.

We note that we only deal with realtime nets with no $\epsilon$-transitions, but our results hold without too much modification when $\epsilon$-transitions are present, as weak simulation over OCNs can be decided in $\PSPACE$ (and $\EXPSPACE$ for a succinct encoding), and the weak simulation pre-order is semilinear as well~\cite{HLMT16}. We considered some model-related variations and concluded that testing the counter for zero freely made checking for history-determinism undecidable.  One could ask about models like reversal bounded one-counter automata~\cite{Iba14}, or automata with bounded number of zero-tests, to gauge the frontier between decidability and undecidability on these systems.

Although not obvious from the main part of the paper, we are confident that our results could easily be extended to safety acceptance conditions. One could also ask, for instance, to look at reachability or B\"uchi and co-B\"uchi acceptance conditions and understand how history-determinism works in these models. 

There are several questions about the expressivity of history-deterministic OCNs which we believe need further study. We have shown that $$\text{DOCN}\subseteq\text{HOCN}\subseteq \text{OCN}\cap \text{DOCA}.$$ An interesting problem would be to prove or disprove if any of these inclusions are strict. In fact, we don't have an example of a language that is accepted by a history-deterministic OCN which is not accepted by a deterministic OCN. 

One could ask similar questions about expressivity of history-determinism in OCAs, i.e. if HOCA = DOCA. Although deciding history-determinism is undecidable, it might be possible for one to show that the language accepted by a history-deterministic OCA is as expressive as deterministic OCA. We remark that the 1-token game $G_1$ characterises history-determinisation for OCAs as well. Moreover, we can again show with similar techniques that if history-deterministic OCAs satisfy the \ssp, then these languages can also be expressed by a deterministic OCA. The key part that we need to prove for determinisation of history-deterministic OCA would be the \ssp. It would be interesting to see how such a proof would look like, given the status of deciding history-determinism being undecidable for OCA.

%% file: Appendix3.tex
\subsection{Simulation Games}\label{appendix:Sim}
\input{Appendix3a}
\subsection{Proof of Lemma~\ref{lemma:G1HD}}
\input{Appendix3b.tex}
\subsection{Proof of Lemma~\ref{lemma:G1Simulation}}
\input{Appendix3c.tex}
\subsection{Proof of Lemma~\ref{lemma:LB}}
\input{Appendix3d.tex}

%% file: Appendix3a.tex
We argue that deciding the winner in simulation game (cf. Definition \ref{def:simulationgame}) is logspace interreducible to deciding the winner in the version of simulation games when the winning condition is given solely by the inability of the either players to choose transitions, and not by accepting states. 

Note that in Definition~\ref{def:simulationgame} for simulation games, we can complete both the one-counter nets by adding a rejecting sink state in each of them, from which we have a transition from every state on $\Sigma$ that does not change the counter value. We also add self loops on $\Sigma$ on the sink state that do not change the counter value. This slight modification does not change the winner in the simulation game. 
Consider the following decision problem, which we call {\scshape simulation}:
\begin{quote}
\textit{Given:} Two complete one-counter nets $\Nc$ and $\Mc$, and configurations $(p,k) \in \Cc(\Nc)$ and $(p',k') \in \Cc(\Mc)$\\
\textit{Question:} Does $\eve$ win the simulation game $\simg{\Nc,(p,k)}{\Mc,(p',k')}$?
\end{quote}

We formally define the OriginalSim game, which is the simulation game where the winning condition is given by the inability of the either player to choose transitions, as defined in literature~\cite{HLMT16}. 

\begin{definition}
Let  $\Nc = (Q, \Sigma,\Delta)$ and $\Mc = (Q,\Sigma',\Delta')$ be two one-counter nets. Given two configuration $(p,k)$ and $(p',k')$ in $\Nc$ and $\Mc$ respectively with $k,k' \in \mathbb{N}$, the OriginalSim game between $\Nc$ and $\Mc$ at position $((p,k),(p',k'))$, is a two player game between $\adam$ and $\eve$, with positions in $\Cc(\Nc) \times \Cc(\Mc)$ where the initial position is $((p_0,k_0),(p'_0,k'_0)) = ((p,k),(p',k'))$. At round $i$ of the play, where the position is: $((p_i,k_i),(p'_i,k'_i))$:
\begin{itemize}
    \item $\adam$ selects a letter $a \in \Sigma$, and a transition $(p_i,k_i) \xrightarrow{a,d} (p_{i+1},k_{i+1})$ in $\Nc$
    \item $\eve$ selects an $a$-transition $(p'_i,k'_i) \xrightarrow{a,d'} (p'_{i+1},k'_{i+1})$ in $\Mc$
\end{itemize}
If after choosing a letter, $\adam$ can't choose a transition, then $\adam$ loses. If after $\adam$ having chosen a transiting, $\eve$ is unable to choose a transition, then $\eve$ loses. Else, the game goes to $((p_{i+1},k_{i+1}),(p'_{i+1},k'_{i+1}))$ for another round of the play. The player $\eve$ wins any infinite play.
\end{definition}
We will also call {\scshape originalsim}, the decision problem of asking if $\eve$ wins the game defined above.
\begin{quote}
\textit{Given:} Two one-counter nets $\Nc$ and $\Mc$, and configurations $(p,k) \in \Cc(\Nc)$, and $(p',k') \in \Mc$, \\ 
\textit{Question:} Does $\eve$ win the OriginalSim game between $\Nc$ and $\Mc$ at position $((p,k),(p',k'))$?
\end{quote}

We now show that the two problems are log-space inter-reducible to each other for asking the decision problem about the winner of the game. 

\paragraph*{Reducing {\scshape Simulation} to {\scshape OriginalSim}:} 
Given an instance of problem {\scshape Simulation}, with two complete OCNs $\Nc = (Q,\Sigma,\Delta,q_0,F)$ and $\Mc = (Q',\Sigma,\Delta',q'_0,F')$ with configurations $(p,k) \in \Cc(\Nc)$ and $(p',k')\in \Cc(\Mc)$ respectively, we reduce it to an instance of {\scshape originalsim}. We construct the (not necessarily complete) net $\Nc'$ ($\Mc'$) by introducing a new alphabet $\$ \notin \Sigma$ to $\Nc$ ($\Mc$), and adding self loops on $\$$ on final states in $\Nc$ ($\Mc$) that do not change the counter. That is, for each state $q\in F$ ($q\in F'$), we introduce the transition $(q,\$,0,q)$. 

We claim that for any configurations $((p,k),(p',k')) \in \Cc(\Nc) \times \Cc(\Mc)$, the player $\eve$ wins the simulation game $\simg{\Cc(\Nc),(p,k)}{\Cc(\Mc),(p',k')}$ if and only if the player $\eve$ wins the OriginalSim game  between $\Nc$ and $\Mc$ at position $((p,k),(p',k'))$. 

$\Rightarrow:$ Suppose $\eve$ plays in the OriginalSim game between $\Nc'$ and $\Mc'$ according to a winning positional strategy in the simulation game between $\Nc$ and $\Mc$ whenever $\adam$ chooses a letter in $\Sigma$. Note that there is at most one $\$$ transition on each state, so $\eve$ either has no choice or a unique choice for choosing a transition in the simulation game at her token.

If $\adam$ never chooses $\$$ in OriginalSim game between $\Nc'$ and $\Mc'$, then $\eve$ wins the game as both $\Nc'$ and $\Mc'$ are complete on $\Sigma$. Now, whenever $\adam$ chooses $\$$ and a transition on $\$$ in $\Nc'$, then $\adam$'s token in $\Nc'$ must have been at a state which is accepting in $\Nc$. As $\eve$ was playing according to her winning strategy, $\eve$'s token in $\Mc'$ would have been at a state corresponding to an accepting one in $\Mc$, which means $\eve$ would be able to take a $\$$-transition as well. Note that taking a $\$$ transition does not change the counter value in both $\Mc'$ and $\Nc'$. Thus, $\eve$ is able to choose a transition whenever $\adam$ can choose one, and hence $\eve$ wins OriginalSim game. 

$\Leftarrow:$ Suppose $\eve$ wins the OriginalSim game between $\Nc'$ and $\Mc'$, and $\eve$ plays in the simulation game according to a winning strategy in the OriginalSim game. Then, whenever $\adam$'s token in $\Nc$ is at a final state in the simulation game, $\eve$'s must be at a final state as well. If not, then $\adam$ would have been able to take a $\$$ transition in the OriginalSim game, while $\eve$ wouldn't be able to, which contradicts the fact that $\eve$ was playing according to a winning strategy,

\paragraph*{Reducing {\scshape originalsim} to  {\scshape Simulation}:} Given an instance of problem {\scshape originalsim}, with two OCNs $\Nc' = (Q,\Sigma,\Delta,q_0)$ and $\Mc' = (Q',\Sigma,\Delta',q'_0)$ with configurations $(p,k) \in \Cc(\Nc)$ and $(p',k')\in \Cc(\Mc)$ respectively, we reduce it to an instance of {\scshape Simulation}. We construct $\Nc$ ($\Mc$) by completing the net by adding transitions on $\Sigma$ that do not change the counter value to a sink state $s$ ($s'$) which is rejecting, and making all the original states $Q$ ($Q'$) accepting. 

We claim that for any configuration $((p,k),(p',k')) \in \Cc(\Nc') \times \Cc(\Mc')$, $\eve$ wins the OriginalSim game between $\Nc'$ and $\Mc'$ at position $((p,k),(p',k'))$ if and only if the player $\eve$ wins the simulation game $\simg{\Cc(\Nc),(p,k)}{\Cc(\Mc),(p',k')}$.

$\Rightarrow:$ Suppose $\eve$ plays in the simulation game between $\Nc$ and $\Mc$ according to a winning strategy in the OriginalSim game between $\Nc'$ and $\Mc'$. Note that once $\adam$'s token goes to the sink state $s$ which is rejecting, then $\eve$ wins the simulation game, as $\adam$ will never see an accepting state. If $\adam$ stays in the states of $\Nc'$ that were also in $\Nc$, then $\eve$ would also be able to stays in the states of $\Mc'$ that were also in $\Mc$. As these states comprises of all accepting states in $\Nc$ and $\Mc$, this implies that whenever $\adam$'s token in $\Nc$ is at an accepting state, so is $\eve$'s.

$\Leftarrow:$ Suppose $\eve$ wins the simulation game between $\Nc$ and $\Mc$, and $\eve$ plays in the OriginalSim game according to a winning strategy for $\eve$ in the simulation game. Then, at any position of a play according to $\eve$'s strategy, whenever $\adam$ is able to take a transition on $a$ in $\Nc'$, then $\eve$ must be able to take a transition in $\Mc'$ according to her strategy as well; if not, then in the simulation game, $\eve$'s token would be at a rejecting state in $\Mc$ while $\adam$'s would be at an accepting state in $\Nc$, which contradicts the fact that $\eve$ was playing according to a winning strategy in the simulation game. 

%% file: Appendix3b.tex
Before proving lemma~\ref{lemma:G1HD}, we make a few observations about the strategies in the letter game. Note that in the letter game, each player winning the game has a positional winning strategy, as it is a reachability game. Suppose $\eve$ wins the letter game, then $\eve$ has a winning strategy which can be given by a (partial) function $$\sigma:(Q \times \mathbb{N}) \times \Sigma^{*} \times \Sigma \rightarrow \Delta^{*}.$$ 
We first prove proposition~\ref{prop:residualresolver} and proposition~\ref{prop:positional}, which we state again below:
\begin{proposition}[also, Proposition~\ref{prop:residualresolver}]
For an OCN $\Nc$, a strategy  $\sigma$ for $\eve$  in the letter game is winning for $\eve$ if and only if  $\sigma$ takes only residual transitions.
\end{proposition}

\begin{proof}

$\Rightarrow$: Let $\sigma : Q \times \mathbb{N} \times \Sigma^* \times \Sigma \rightarrow \Delta$ be a winning strategy for $\eve$ in the letter game. Suppose, $\sigma(q,k,u,a) = (q,k) \xrightarrow{a,d} (q',k')$, for some configuration $(q,k)$, reached upon reading a prefix $u$ following $\sigma$. We need to show that the transition $(q,k)\xrightarrow{a,d} (q',k')$ is residual, i.e. $L(q',k') = a^{-1} L(q,k)$. But if there is a word $w \in a^{-1}L(\Nc,(q,k))\setminus L(\Nc,(q',k'))$, then the strategy $\sigma$ can't end at an accepting state on the word $uaw$, but $uaw \in L$, a contradiction. 

$\Leftarrow$: Suppose $\sigma$ is an $\eve$ strategy which only takes residual transitions. Then for each word $w$, if $(q,k)$ is the configuration reached upon reading the word $w$, then  $L(\Nc,(q,k)) = w^{-1}L$. If $w \in L$, then $\epsilon \in w^{-1}L$, and hence the configuration $(q,k)$ is accepting. Thus, $\sigma$ is at an accepting state whenever the word read so far is accepting, and thus $\sigma$ is a winning strategy.
\end{proof} 

Using  Proposition~\ref{prop:residualresolver}, we can show that $\eve$'s strategy only depends on the configuration, 
and is independent of the word read so far. Proposition~\ref{prop:positional} below show that we can have a resolver based on the current state and counter value alone. 
\begin{proposition}[also, Proposition~\ref{prop:positional}]
If $\eve$ wins the letter game on an OCN $\Nc$, then $\eve$ has a winning strategy $\sigma$ that only depends on the current configuration of the play, i.e $\sigma$ is a partial function $\sigma: (Q \times \mathbb{N}) \times \Sigma \rightarrow \Delta^{*}$    
\end{proposition}
\begin{proof}
Let $\sigma': (Q \times \mathbb{N}) \times \Sigma^{*} \times \Sigma \rightarrow \Delta^{*}$ be any winning strategy. We define a strategy $\sigma$ as follows: For each configuration $(q,k)$ in the one-counter net $\Nc$, we let $\sigma((q,k),a) = (q,k) \xrightarrow{a,d} (q',k')$, where for some word $u$ on which $\sigma'$ reaches the configuration $(q,k)$, we have $\sigma'((q,k),u,a) = (q,k) \xrightarrow{a,d} (q',k')$. By the Proposition ~\ref{prop:residualresolver}, $\sigma'$ takes only residual transitions, implying $\sigma$ takes only residual transitions as well, and it follows from proposition ~\ref{prop:residualresolver}  that $\sigma$ is a winning strategy for $\eve$ which depends only on the configuration.  \end{proof}

Having characterised what strategies look like in the letter game, we are finally equipped to prove Lemma~\ref{lemma:G1HD} using Proposition~\ref{prop:residualresolver} and Proposition~\ref{prop:positional}.

\paragraph*{Proof of lemma \ref{lemma:G1HD}}
Let $\gamma$ be a winning strategy for $\eve$ in $G_1$, and let $\lambda$ be the strategy in the letter game, derived from $\gamma$ where $\adam$ copies $\eve$'s play. We show that $\lambda$ takes only residual transitions. This is enough because of Proposition~\ref{prop:residualresolver}.  Assume to the contrary, that $\eve$ takes a non-residual transition in a play following $\gamma$, $(p,k) \xrightarrow{a,d} (p',k')$ at the position $((p,k),(p,k))$ in $G_1$. As $\adam$'s token is also at $(p,k)$, $\adam$ can win by constructing a word $av \in \Lc(p,k)$, such that $v \notin \Lc(p',k')$ so that $\eve$ can't produce an accepting run of $v$ from $(p',k')$, while constructing an accepting run for $a\cdot v$ from $(p,k)$ with his token. Thus, $\eve$ loses $G_1$, and hence $\gamma$ cannot be a winning strategy in $G_1$, a contradiction. 

%% file: Appendix3c.tex
We state the construction's intuition as well as give a rigorous construction of the nets $\Mc$ and $\Mc'$ side-by-side for ease of reference.
Let $\Nc = (Q,\Sigma,\Delta,q_0,F)$ be a one-counter net. We assume the net $\Nc$ to be complete, as we can add a rejecting sink state, from which we have a transition from every state on $\Sigma$ that does not change the counter value. This modification does not change the winner in the letter game of $\Nc$. 

\paragraph*{Construction of $\Mc'$}
We define $\Mc$ to essentially be the net $\Nc$, but containing an extra state $q_\#$, and a newly added letter $\#$. Along with this, there are new transitions added. These transitions loop on the new state $q_{\#}$ for any letter in the expanded alphabet, without changing the counter. Transitions are added to reach this state $q_\#$ by reading $\#$ from a final state in $\Nc$.
More formally, $\Mc' = (Q',\Sigma_{\#},\Delta',q_0,F')$, where 
\begin{itemize}
    \item the set of states is $Q' = Q \cup \{q_{\#}\}$, $\Sigma_{\#} = \Sigma \cup \{\#\}$,
    \item the accepting state is just $F'=\{q_{\#}\}$, and
    \item the set of transitions $\Delta'$ is defined as $$\Delta'= \Delta \cup \{(q,\#,0,q_\#) \mid q\in F\} \cup \{(q_{\#},a,0,q_{\#}) \mid a \in \Sigma_{\#}\}.$$ 
\end{itemize}
Note that $\Mc'$ has the same initial state $q_0$, as $\Nc$. 

\paragraph*{Construction of $\Mc$}
We construct $\Mc$ to contain approximately $(|\Sigma|+1)$ copies of the states in $\Nc$. These copies help remember the previous letter read in the state space of the OCN, and mimic a 'one-step lag' in $\Nc$. The transitions of $\Mc$ on reading a letter, store the letter in the state space. However, in the projection of the $\Nc$ component  of $\Mc$'s state, the transition is based on the letter that was previously stored in the letter component as opposed to the current letter read, which will now be stored in the state space of the new state in $\Mc$. This is built to capture a play of $G_1$ by creating a delay in the simulation game for $\adam$ by forcing a one-step delay during his play. 
We formalise this below by defining $\Mc = (Q_M,\Sigma_{\#},\Delta_M,s,F_M)$, where 
\begin{itemize}
    \item $Q_M = (Q \times \Sigma) \cup \{s\} \cup (F \times \{\#\})$,
    \item  $s$ is the initial state,
    \item $F_M$, the set of final states is $F \times \{\#\}$, and 
    \item the set of transitions $\Delta_M$ is the union of the following sets: 
\begin{itemize}
    \item $ \{(s,a,0,(q_I,a)) \mid {a \in \Sigma}\}$
    \item $\{((p,a),b,d,(q,b)) \mid (p,a,d,q) \in \Delta\}$
    \item $\{\left((p,a),\#,d,(q,\#)\right) \mid (p,a,d,q) \in \Delta \ \text{and} \ q\in F \}$
\end{itemize}
\end{itemize}  
We now prove that $\eve$ wins the game $G_1$ on $\Nc$ if and only if $\eve$ wins the simulation game $\Gc(\Mc \simulates \Mc')$. We first define a slightly different linearisation of $G_1$. We modify $G_1$ so that $\adam$ has to wait one more turn to execute a transition in his token. For the first position where both $\eve$'s and $\adam$'s tokens are at the initial configuration:
\begin{enumerate}
    \item $\adam$ picks a letter $a\in\Sigma$.
    \item $\eve$ responds by picking a transition $(q_I,0)\xrightarrow{a,k'_1} (q_1,k'_1) \in \Delta$.
\end{enumerate}
Now, we say that the token is at position $(q_1,k_1')$ for $\eve$ and $(q_I,0)$ with $\adam$ having to execute $a$.
For the $i^{th}$ turn, from configurations $(q_i,k_i')$ for $\eve$ and $(p_{i-1},k_{i-1})$ for $\adam$ with him having to execute a letter $a_{i-1}\in\Sigma$,
\begin{enumerate}
    \item $\adam$ picks a transition $(p_{i-1},k_{i-1})\xrightarrow{a_{i-1},d}(p_{i},k_{i})$ as well as a letter $a_{i}\in\Sigma$.
    \item $\eve$ responds by picking a transition $(q_{i},k_i')\xrightarrow{a_i,d'} (q_{i+1},k_{i+1}')$ on $a_i$.
\end{enumerate}
The winning condition for $\eve$ is the following: If after $i$ rounds for each $i \in \mathbb{N}$, the player $\adam$ is at a configuration $(p_{i-1,k_{i-1}})$ such that $\adam$ can pick the transition $(p_{i-1},k_{i-1})\xrightarrow{a_{i},d}(p_{i},k_{i})$ with $p_i\in F$, then it must be that $q_i$ was already a final state (Recall that we assumed without loss of generality that the net $\Nc$ is complete). 

It's easy to see that the the above modified formalisation of $G_1$ is equivalent to the standard $1$-token game. This modified formalisation however, would make it easier to see the equivalence with the simulation game of the two OCNs constructed. 

\paragraph*{If $\eve$ wins $G_1$, then $\eve$ wins $\Gc(\Mc \simulates \Mc')$}
Suppose $\eve$ wins the modified $G_1$ in $\Nc$. Let $\sigma$ be a winning strategy for $\eve$ in the modified game $G_1$. Then, in the simulation game $\Gc(\Mc \simulates \Mc')$, $\eve$ can win by inductively constructing a simultaneous play of modified $G_1$ which follows $\sigma$ in her memory, in order to choose transitions in the simulation game. 

\begin{itemize}
    \item The simulation game starts at $((s,0),(q_I,0))$. For any letter $a \in \Sigma$ that $\adam$ picks (if he picks $\#$, he loses, as he can't move from $\#$ on $s$), the transitions available take him to $(q_I,a)$, since this transition is deterministic.
    To respond to the above play, $\eve$ uses the winning strategy of modified $G_1$. In $\eve$'s view she would respond with the transition that she would have in the modified $G_1$ if $\adam$ picked $a\in\Sigma$, using her strategy $\sigma$.  
    Suppose $\eve$'s strategy $\sigma$ in the modified $G_1$ picks the transition $(q_I,0) \xrightarrow{a,k_1} (q_1,k_1)$. Then she uses this as her strategy to choose the corresponding transition in $\Mc'$ in the simulation game. She also builds a play in the modified $G_1$ game where she has made a move from the initial state to $(q_1,k_1)$, with $\adam$ yet to execute an $a$-transition for his token. The configuration in the modified $G_1$, which is stored in $\eve$'s memory is at $((q_1,k_1),(q_I,0))$, with $\adam$ yet to execute an $a$-transition from $(q_I,0)$.
    \item Suppose, $\adam$ has not not played $\#$ in the first $i$ rounds of the simulation game. Let the simulation game be at the position $(((p_{i-1},a),k),(q_{i-1},k'))$.
    
    Then inductively, the corresponding run of $\eve$ in the modified $G_1$ for such a play in $\Nc$ is going to be $((q_i,k_i'),(p_{i-1},k_{i-1}))$ where $\adam$ is yet to choose a transition on $a$ from the configuration $(p_{i-1},k_{i-1})$. In the simulation game $\Gc(\Mc \simulates \Mc')$,
    \begin{itemize}
        \item  $\adam$ chooses a letter $b$, and a corresponding transition $((p_{i-1},a),b,d,(p_{i},b))$ from the available transitions of $\Mc$.
        \item $\eve$ responds as though, $\adam$ executed the transition $(p_{i-1},k_{i-1}) \xrightarrow{a,d}(p_{i},k_i)$ and picked the letter $b$ in the modified game $G_1$ from the tuple of configurations $((q_i,k_i'),(p_{i-1},k_{i-1}))$ to the tuple of configurations $((q_i,k_i'),(p_{i},k_i))$.
        
        The player $\eve$'s strategy $\sigma$ in the modified $G_1$, would have prescribed a transition, say $(q_i,k_i')\xrightarrow{b,d'} (q_{i+1},k_{i+1}')$. She picks the same transition available to her in $\Mc'$ as a response in the simulation game. 
    \end{itemize}
    After the transition $(q_i,k_i')\xrightarrow{b,d'} (q_{i+1},k_{i+1}')$ was picked, the corresponding inductive game of moidifed $G_1$ built is updated to the tuple of configurations $((q_{i+1},k_{i+1}'), (p_{i},k_i))$ with $\adam$ yet to pick a transition on $b$.
    \item Suppose $\adam$ picks $\#$ at some position $i$, and suppose the configuration of the tokens were at  $((p_{i-1},a),k_{i-1}),(q_i,k_i')$, then $\adam$ can make a move if and only if $\adam$ can get to a final state, by a transition $(p_{i-1},a,d,p_{i+1})$, for $p_{i+1} \in F$. But then, as $\eve$ was playing according to a winning strategy $\sigma$ in the modified game $G_1$, her token in the modified $G_1$ must have been at a final state as well. This enables $\eve$ take a transition to $q_{\#}$ in the simulation game as well, from which she can win.
\end{itemize} 
Thus, we have shown that if $\eve$ wins the modified $G_1$, then $\eve$ wins the simulation game $\Gc(\Mc \simulates \Mc')$ as well. 

\paragraph*{If $\eve$ wins $\Gc(\Mc \simulates \Mc')$, then $\eve$ wins $G_1$:}
We will now show the other direction that if  $\eve$ wins the simulation game $\Gc(\Mc \simulates \Mc')$, then $\eve$ wins the modified $G_1$ over $\Nc$. Let $\tau$ be a winning strategy for $\eve$ in the simulation game. Then, we will show that $\eve$ can win the modified $G_1$ over $\Nc$ by inductively constructing a simultaneous play of simulation game which follows $\tau$ in her memory in order to choose transitions in the modified $G_1$ game. Formally, $\eve$ plays in the modified $G_1$ using the strategy $\tau$ as follows:

\begin{itemize}
    \item The modification of the game $G_1$ starts while all the tokens for each of the player both correspond to the configurations $(q_I,0), (q_I,0)$. In $G_1$, 
    \begin{itemize}
        \item For any letter $a \in \Sigma$ picked by $\adam$, we consider the corresponding play of $\adam$ in the game $\Gc(\Mc\simulates\Mc')$ here: \begin{itemize}
           \item $\adam$ picks a transition from $(s,0)\xrightarrow{a,0} ((q_I,a),0)$. In the simulation game,
           \item $\eve$ responds with a transition of $\Mc'$ suggested by her winning strategy $\tau$ with the transition $(q_I,0)\xrightarrow{a,k_1'}(q_1,k_1')$.
        \end{itemize} 
    \item  For the game $G_1$, player $\eve$ is to respond with the corresponding transition $(q_I,0)\xrightarrow{a,k_1'}(q_1,k_1')$ above in the net $\Nc$ during her turn. So the configuration in the modified $G_1$, after the first round is: $\eve$'s token is at $(q_1,k_1')$, while $\adam$'s token is at $(q_I,0)$, waiting to make a move on $a$.
    \end{itemize}
    \item     In the $i^{th}$ turn of the modified game $G_1$, suppose the game is at configuration $(q_i,k_i')$ for $\eve$ and $(p_{i-1},k_{i-1})$ for $\adam$, with him having to execute a letter $a_{i-1}$ in $\Sigma$. Then, the inductive run of the simulation game would be at the following position: $(q_i,k_i')$ in the net $\Mc'$ and $((p_{i-1},a_{i-1}),k_{i-1})$ in the net $\Mc$ for this round of the simulation game:
    \begin{itemize}
        \item In the modified $G_1$, $\adam$ picks a transition over $(p_{i-1},k_{i-1})\xrightarrow{a_{i-1},d} (p_{i},k_{i})$ over $\Nc$, and a letter $a_{i+1}$.
        This corresponds to a unique transition over the letter $a_i$ in $\Mc$, as $((p_{i-1},a_{i-1}),k_{i-1})\xrightarrow{a_i,d} ((p_i,a_i),k_{i})$.
        We assume that $\adam$ in the simulation game extends the run with the above pair of letter and transition.
        \item The player $\eve$ responds with a transition $(q_i,k_i')\xrightarrow{a_i,d'} (q_{i+1},k_{i+1}')$ from the net $\Mc'$ in the simulation game, using $\tau$. 
    \end{itemize}
    \end{itemize}
    The same transition is picked by $\eve$ for modified $G_1$: $(q_i,k_i')\xrightarrow{a_i,d} (q_{i+1},k_{i+1}')$, This transition is available in $\Nc$ by construction. Therefore, the modified game $G_1$ is now at configurations: $(q_{i+1},k_{i+1}')$ for $\eve$ and $(p_i,k_i)$ for $\adam$, who is waiting to make a move on $a_{i}$. The corresponding simulation game is at the configuration $((p_{i},a_{i}),k_{i})$ for $\adam$ and  $(q_{i+1},k_{i+1}')$ for $\eve$.
    
    We now argue that the above described strategy is winning for $\eve$, in the modified $G_1$. Consider a configuration in the round $i$ of the play, where $\eve$'s token is at $(q_i,k_i')$ and $\adam$'s token is at $(p_{i-1},k_{i-1})$,with him having to execute a letter $a_{i}\in\Sigma$. It suffices to show that if $\adam$ can end up in an accepting state after executing an $a_{i}$ transition from $(p_{i-1},k_{i-1})$, then the state of $\eve$'s configuration $q_i$ must also be an accepting state in $\Nc$.
    
    This follows from the fact that strategy $\tau$ in the simulation game was winning for $\eve$. At the $i^{th}$ round of the play, the inductively built run in the simulation game would be at a configuration $(q_i,k_i')$ in the net $\Mc'$ and $((p_{i-1},a_{i}),k_{i-1})$ in the net $\Mc$. If $\adam$ could have picked a transition to an accepting configuration on reading $a_i$ in the net $\Nc$ from $(p_{i-1},a_{i})$, then $\adam$ has a transition on $\#$ enabled for his play in $\Mc$. The player $\eve$ would then must be able to respond with a $\#$-transition (as $\tau$ was a winning strategy), but that is only available from states $q_i$ which are accepting. 

%% file: Appendix3d.tex
\begin{proof}
The proof goes by reducing from the problem of checking non-emptiness  of alternating finite-state automata over a unary alphabet. This variation of the problem was shown to be $\PSPACE$-complete by Holzer~\cite{Hol95}, with its proof simplified by Jan\v{c}ar and Sawa~\cite{JS07}.

We define an alternating finite-state automata $\Ac$ over the unary alphabet as follows: $\Ac = (Q = Q_{\lor}\uplus Q_{\land}, \delta,q_0, F)$ where $Q$ is a finite set of states, partitioned among two players $\lor$ and $\land$, $q_0\in Q$ is the start state, $F\subseteq Q$ is the final state and the transitions are $\delta\subseteq (Q_{\lor}\times Q_{\land})\cup (Q_{\land} \times Q_{\lor})$. 

The empty-string $\epsilon$, which has length 0 is in the language accepted by $\Ac$ from $q\in Q$ iff $q\in F$. 
We say a word of length $n\geq 1$ is accepted from $q$ if either 
\begin{itemize}
    \item $q\in Q_{\lor}$ and there exists $(q,q')\in\delta$ such that a word of length $n-1$ accepted by $\Ac$ from $q'$; or 
    \item $q\in Q_{\land}$ and for every $(q,q')\in\delta$, there is a word of length $n-1$ accepted by $\Ac$ from $q'$.
\end{itemize}
A word of length $n$ is accepted by $\Ac$ if it is accepted from the initial state $q_0$.

Intuitively, this can be thought of as a game between two players, $\lor$ and $\land$ where an $n$ length word is accepting in the automaton if and only if $\lor$ wins in the $n$-length play in the above automaton viewed as a reachability game to one of the final states. The player $\land$'s objective is adversarial to $\lor$.

Given such an alternating finite automaton, we construct a one-counter net $\Hc$ that is history-deterministic if and only if $\Ac$ is empty.

We would like to show that 
\begin{itemize}
    \item $\adam$ wins the letter game on $\Hc$ if there is a word accepted by $\Ac$
    \item $\eve$ wins the letter game on $\Hc$ when $\Ac$ is empty
\end{itemize}

The idea is that we add a state $q_I$, which will be a new initial state of the one-counter net. On $q_I$, $\adam$ can read a special input $\one$, which will increase the counter value while reading this input. To win, $\adam$ would have to read as many $\one$s as the length of an accepting word, and later prove that indeed this word is accepting. While constructing the run, since $\adam$ is the one proving non-emptiness, he will be resolving what we would think of as `existential' choices, here denoted by transitions of $\lor$. The player $\eve$ on other hand, would be resolving the `universal' choices which are the transitions of $\land$. 

A run is constructed by having a copy of states and encoding in the alphabet, the choices to be made by $\lor$ player so that $\adam$ can pick the letter. For the $\land$ player, we want $\eve$ to make the 'universal' choice, so we encode this in the non-determinism. But to ensure $\adam$ ensures eve to `fairly' pick the choices, if $\adam$ decides to read a word state that does not respect the current state that $\eve$ is in, she can move to a state from which there is no non-determinism resolution, and $\eve$ can accept any valid suffix. 

If $\adam$ has reached a final state with counter value exactly 0, then he has displayed that such an accepting run exists, from where $\adam$ reads a $\$$, and can win the letter game. But if this run has reached a final state with a positive value, and $\adam$ reads a $\$$, then $\eve$ can win the letter game. This is done by a gadget described as follows:

From any state, $\adam$ can read a symbol $\$$. If the counter is non-zero, or if the state is non-accepting, $\eve$ has a transition that subtracts 1 from the counter and goes to a state from which any of the two special symbols are accepted: $\heartsuit$ and $\clubsuit$. 
Whereas, if the same symbol $\$$ is read by  $\adam$ when the counter is empty and at an accepting state,then the only transitions enabled have a non-determinism that cannot be resolved by $\eve$, where she would have to predict if $\heartsuit$ or $\clubsuit$ will be seen in the future.
\begin{figure}
    \centering
    \makebox[\textwidth][c]{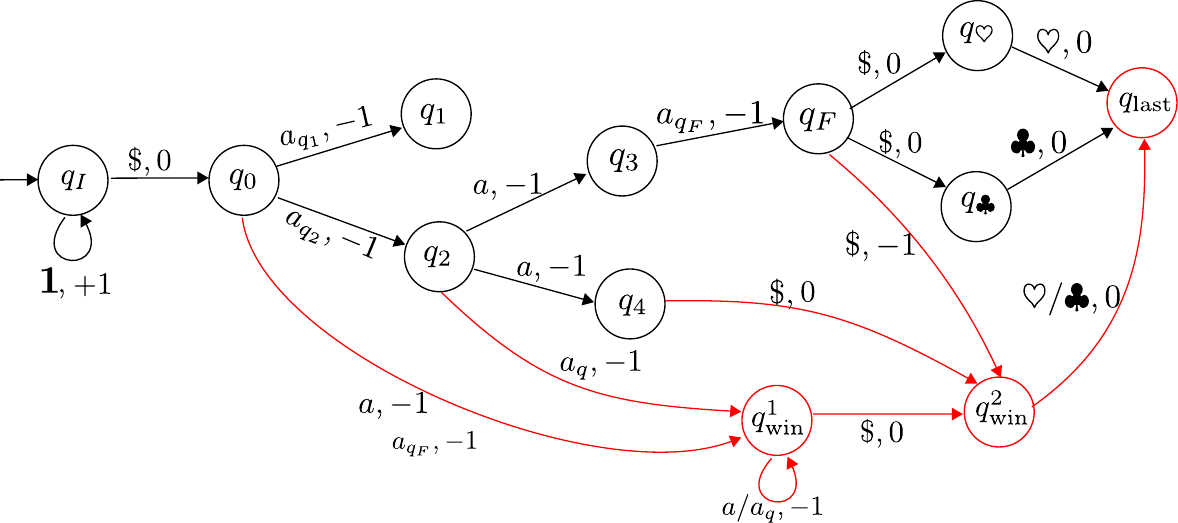}
     \caption{Some of the states in the net $\Hc$ constructed. Here, we assume that $q_0,q_3,q_4\in Q_\land$ whereas $q_1,q_2,q_F\in Q_\lor$. All states are accepting. Observe that at $q_F$, a final state, the non-determinism is history-deterministic only when the counter is non-empty. Also notice that the red arrows/states are when $\adam$ proposes a letter that does not correspond to a feasible transition played by $\eve$ in the AFA word resolution.~\label{fig:PSPACEhard}}
\end{figure} 
We refer the reader to a pictorial representation of the construction in Figure~\ref{fig:PSPACEhard}.

We define the OCN $\Hc = (Q_H, \Sigma, \Delta_H, q_I, F_H)$ where
\begin{itemize}
    \item $Q_H = Q \cup \{q_I\}\cup \{q_{\clubsuit},q_{\heartsuit}\}\cup\{ q_{\win}^1,q_{\win}^2\}\cup \{q_{\last}\}$,
    \item $\Sigma = \{\$,\heartsuit,\clubsuit,\one\}\cup \{a_{q} \mid  q\in Q_{\land}\}\cup \{a\}$, 
    \item $F_H =  Q_H$, where all states defined are final states, and
    \item $\Delta_H$, the set of transitions are the union of the following sets given below:
    \begin{enumerate}
        \item $\{(q_I, \one,1 , q_I), (q_I, \$,0,  q_0)\}$
        \item $\{(q_{\lor}, a_q,-1,q)\mid (q_{\lor}, q)\in \Delta \text{ and }q_{\lor} \in Q_{\lor}\}$
        \item $\{(q_{\lor},a_q,-1, q_{\win}^1)\mid (q_{\lor},q)\notin \Delta \text{ and } q_{\lor}\in  Q_{\lor}\}$
        \item$\{(q_{\lor},a,-1, q_{\win}^1)\mid q_{\lor}\in Q_{\lor}\}$
        \item $\{(q_{\land},a,-1, q)\mid (q_{\land},q)\in \Delta\text{ and }q_{\land}\in Q_{\land}\}$
        \item $\{(q_{\land},a_p,-1, q_{\win}^1)\mid p\in Q\text{ and }q_{\land}\in Q_{\land}\}$
        
        \item $\{(q_F,\$,-1,q_{\win}^2), (q_F, \$, 0, q_{\clubsuit}), (q_F, \$, 0, q_{\heartsuit})\mid q_F\in F\}$
        \item $\{(q,\$,0,q_{\win}^2) \mid q \notin F\}$
        \item $\{(q_{\win}^1, \$, 0, q_{\win}^2)\}$
        \item $\{(q_{\win}^1, b, -1, q_{\win}^1)\mid b = a_q \text{ or }b=a\}$
        \item $\{(q_{\heartsuit}, \heartsuit,0,q_{\last}), (q_{\clubsuit}, \clubsuit,0,q_{\last}), (q_{\win}^2, \heartsuit, 0, q_{\last}), (q_{\win}^2, \clubsuit, 0, q_{\last})\}$. 
    \end{enumerate}
\end{itemize}

The state space consists of a state $q_I$, and a copy of the states $Q$ of $\Ac$.
Moreover, there are states $q_{\win}^1$ and $q_{\win}^2$ from which intuitively, all words that are in the language of $\Hc$ henceforth are accepted. These states will be winning for $\eve$ if she reaches them in the letter game, and hence $\adam$ loses from that state. 
There are also two states $q_{\heartsuit}$ and $q_{\clubsuit}$ which ideally $\adam$ would like $\eve$ to reach in the letter game.
Intuitively, if $\eve$ reaches a copy of $Q$'s final states with 0 in the counter, she looses as she has to pick between the states $q_{\heartsuit}$ and $q_{\clubsuit}$ and this kind of  non-determinism makes it winning for $\adam$. 

At state $q_I$, $\adam$ can read as many $\one$s as he wants, which increases his counter value. Alternatively, there is a symbol $\$$ that he can read when the play moves from $q_I$ to $q_0$.

This signals that $\adam$ is ready to display that he can construct a branch of the run-tree with length exactly the counter value. Once he has finished this run construction, he can again  use $\$$ to signal the end of a run-constructed at the states, producing the non-determinism that makes $\eve$ win iff the counter is non-zero, or if the state is not accepting.

The letters $a_{q}$ enables transitions from states which belong to $\lor$, whereas $a$ enables arbitrary non-deterministic choice consistent with the original automaton's transitions from states belonging to $\land$. 
From $q$, the letter $a_q$ enables transitions that are deterministic.

All states in the net $\Hc$ are final, but the net $\Hc$ is not complete, making it non-universal. 

The transitions are mostly as explained before, but we supply some additional discussion to understand better. At the initial state, $\adam$ could read a counter value and increase arbitrarily. After this he can read $\$$ eventually and enter $q_0$.
On entering $q_0$, the initial state in $\Ac$, $\adam$ resolves choices of the player $\lor$ which can be thought of as `existential' choices whereas and $\eve$ the `universal' choices, the choices of $\land$. This is done by encoding this in the alphabet and non-determinism respectively. While each choice is made, the counter value is decreased as a count-down to the length of the word.

We would like to emphasise here that we add transitions to state $q_{\win}^1$ if $\adam$ reads a letter that does not extend the transitions picked by $\eve$'s run constructed. 
This enables $\eve$ to pick  her non-deterministic transition in such a way that $\adam$ cannot ensure her loss trivially. This is done by adding transitions from every state in the copy of $Q$, to a state $q_{\win}^1$ for letters that are such that they do not extend a transition picked by eve while resolving non-determinism. This is made more precise in the definitions of the transitions. If he does pick a correct run then $\eve$ can also only construct a run, and cannot reach $q_{\win}^1$. 

There are several transitions on $\$$. The idea here is that $\adam$ can read $\$$ once he is at a final state with counter 0. Note that he can also read $\$$ at a non-final state, but then $\eve$ can pick a transition to $q_{\win}^2$, from where $\eve$ can win the letter game. Similarly, 
if he does read $\$$ from a final state when the counter value is not zero, transitions are enabled for $\eve$ that reach $q_{\win}^2$. 
Finally, if the counter value is indeed zero, then the only two transitions enabled make $\eve$ pick in advance for going to state $q_{\heartsuit}$ and $q_{\clubsuit}$. From these two states there is only one transition $\heartsuit$ and $\clubsuit$ respectively.

At $q_{\win}^1$ however, $\eve$ can reach accepting state on any series of letters $a$ or $a_{q}$, whilst decrementing the counter, then seeing a $\$$ and then read $\heartsuit$ or $\clubsuit$ with no non-determinism. 

Note that the language accepted by $\Hc$ is the prefix closure of $$\{\one^n \$ \cdot a_1\cdot a_2\cdot \dots a_k \cdot\$\cdot\{\heartsuit,\clubsuit\}\mid a_i = a \text{ or }a_i = a_q\text{ for some }q\in Q, k\leq n\}$$

\paragraph{Proof of correctness of the construction}
We now proceed to showing that the constructed automaton indeed satisfies the following:
\begin{itemize}
    \item[$\Rightarrow$] $\adam$ wins the letter game on $\Hc$ if there is a word accepted by $\Ac$
    \item[$\Leftarrow$] $\eve$ wins the letter game only when $\Ac$ is empty
\end{itemize}
\paragraph*{$\adam$ wins the letter game on $\Hc$ if $\Ac$ is non-empty:}
We give a strategy for $\adam$ in the letter game:
If $\Ac$ is non-empty, there is some $n$ for which there is an $n$-length word that is accepted by $\Ac$. Without loss of generality, we assume that $n>0$.

$\adam$ reads the letter $\one$  $n$-many times. There is no non-determinism for $\eve$ resolve in this game so far. 
After this, $\adam$ reads $\$$ and the game moves to $q_0$, the copy of the initial state of $\Ac$.

On reaching the copies of the states of $\Ac$, the letter game proceeds following the invariant 
\begin{quote}
    the run constructed so far by $\eve$ in the letter game is at a state $q$, such that there a word accepted by $\Ac$ of length equal to the current counter value from state $q$.
\end{quote}

This is indeed true at the vertex $q_0$, by assumption that there exists a word of length $n$ accepted from $\Ac$, and the counter value has $n$ in it.

Let the current counter value be $k>0$, and the current state be $q$. Then the following $\adam$-strategy preserves the above invariant:
\begin{itemize}
    \item If the play is at a state $q = q_{\lor}\in Q_{\lor}$, then the letters that can be read are of the form $a_q'$ for $q' \in Q_{\land}$.
    Let $(q_{\lor}, q') \in \Delta$ be such that there is a word of length $k-1$ accepted from $q'$. In this case, $\adam$ reads such an $a_{q'}$ as his next transition, leaving $\eve$ with no non-determinism to resolve.
    \item If the play is at a state $q = q_{\land}\in Q_{\land}$, then $\adam$ reads the letter $a$, which lead to an other state $q'$ chosen by $\eve$ such that there is a transition $(q_{\land},q')$ among the transition of $\Ac$. 
    Since there was a transition $(q_{\land},q')$ in the original automaton, it must be the case that there is a word of length $k-1$ accepted from the copy of the state $q'$. No matter how $\eve$ resolves the non-determinism on $a$, she ends up at a state that satisfies the above invariant.
\end{itemize}
Finally, once the game is at a state with an empty counter value, we know from the above invariant that we are at a final state. From there, $\adam$ reads the letter $\$$. Since the counter value is $0$ and the state is a final state, the only transition that are available for $\eve$ are $(q_F,\$,0,q_\heartsuit)$ and $(q_F,\$,0,q_\clubsuit)$, by construction. No matter which of these transitions $\eve$ picks, $\adam$ can respond by picking the other letter not corresponding to the state $\eve$ is at, and win the game. 

\paragraph*{$\eve$ wins the letter game on $\Hc$ if $\Ac$ is empty:}

If the game stays at $q_I$ forever, $\eve$ wins automatically. If not, after reading a sufficient number of $\one$s, $\adam$ chooses letter $\$$. This moves the game to $q_0$. Suppose in this run, $\adam$ enters with $n$ as the counter value. 
Since this automaton accepts no letters by assumption, there is no accepting run of length $n$ from $q$.

Player $\eve$ uses a strategy that follows the following invariant:
\begin{quote}
     if the letter game is at a state $q\in Q$ for $\eve$, then there is no word of length equal to the counter value accepted by $\Ac$ at $q$
\end{quote}
Again, it is true at $q_0$.  If this invariant is true,  when counter value is 0, then the state is not a final state and $\adam$ has to read a $\$$ to ensure that word is still in the language, as he loses immediately otherwise.  But $\eve$ can take then the transition $(q,\$,0,q_{\win}^1)$ to $q_{\win}^1$, and then read $\heartsuit$ or $\clubsuit$. If $\adam$ reads a $\$$ when the counter value is greater than $0$, $\eve$ can again take a transition to $q_{\win}^1$, and win the letter game. 

Now we prove the invariant. 
\begin{itemize}
    \item If the play is at a state $q_{\lor}\in Q_{\lor}$, then no matter what letters $\adam$ proposes,  there is no non-determinism to resolve for $\eve$. If
    \begin{itemize}
        \item $\adam$ reads an $a$ or any $a_{q'}$ such that $q'$ is not adjacent to $q$, then  $\eve$ moves to $q_{\win}^1$;
        \item $\adam$ reads a $\$$, then  $\eve$ moves to $q_{\win}^2$; 
        \item $\adam$ reads $a_{q'}$ with $(q_{\lor},q')$ being a transition in $\Ac$, then the play moves to $q'$, on subtracting $1$ but this is a state from which there is no run of length $k-1$, preserving the invariant.
    \end{itemize}
    \item If the play is at a state $q = q_{\land}\in Q_{\land}$, then the letters that can be read that are of the form $a$, which lead to another state $q'$ such that there is a transition $(q_{\land},q')$ among the transitions of $\Ac$. If $\adam$ reads anything that is of the form $a_{q'}$, then $\eve$ goes to $q_{\win}$. But if not, since there is at least one  transition $(q_{\land},q')$ in the original automaton such that there are no words of length $k-1$ accepting from such a state $q'$, $\eve$ picks that transition in the letter game continuing her play.
\end{itemize}

This shows that $\Hc$ is history-deterministic if and only if $\Ac$ is empty. 
\end{proof}

%% file: AFALB.pdf_tex
\begingroup%
  \makeatletter%
  \providecommand\color[2][]{%
    \errmessage{(Inkscape) Color is used for the text in Inkscape, but the package 'color.sty' is not loaded}%
    \renewcommand\color[2][]{}%
  }%
  \providecommand\transparent[1]{%
    \errmessage{(Inkscape) Transparency is used (non-zero) for the text in Inkscape, but the package 'transparent.sty' is not loaded}%
    \renewcommand\transparent[1]{}%
  }%
  \providecommand\rotatebox[2]{#2}%
  \newcommand*\fsize{\dimexpr\f@size pt\relax}%
  \newcommand*\lineheight[1]{\fontsize{\fsize}{#1\fsize}\selectfont}%
  \ifx\svgwidth\undefined%
    \setlength{\unitlength}{339.17071619bp}%
    \ifx\svgscale\undefined%
      \relax%
    \else%
      \setlength{\unitlength}{\unitlength * \real{\svgscale}}%
    \fi%
  \else%
    \setlength{\unitlength}{\svgwidth}%
  \fi%
  \global\let\svgwidth\undefined%
  \global\let\svgscale\undefined%
  \makeatother%
  \begin{picture}(1,0.44383089)%
    \lineheight{1}%
    \setlength\tabcolsep{0pt}%
    \put(0,0){\includegraphics[width=\unitlength,page=1]{AFALB.pdf}}%
  \end{picture}%
\endgroup%

%% file: Appendix4.tex
\subsection{Proof of Lemma~\ref{lemma:DetBySSP}}
\input{Appendix4a.tex}
\subsection{Proof of Lemma~\ref{lemma:netSSP}}
\input{appendix4b.tex}
\subsection{Proof of Lemma~\ref{lemma:HDinclusion}}
\input{Appendix4c.tex}
\subsection{Proof of Lemma~\ref{lemma:universality}}
\input{Appendix4d.tex}

%% file: Appendix4a.tex
\begin{proof}
We assume the history-deterministic OCN $\Nc$ is such that it satisfies \ssp. Suppose, for each transition $\delta$, the set $\Sc_{\delta}$ is an eventually periodic set with its period as $P_{\delta}$, with the maximum number in the pre-periodic part as $I_{\delta}$. Let $I = \max \{I_{\delta}\}_{\delta \in \Delta}$, and $P = \prod_{\delta \in \Delta} \{P_{\delta}\}$. Thus, each set $\Sc_{\delta}$ for each transition $\delta$ can be expressed an eventually periodic set with period $P$, and all numbers in the pre-periodic part at most $I$. 

 We first construct a non-deterministic one-counter automata $\Bc$ that accepts the same language as $\Nc$. Intuitively, the automaton $\Bc$ is constructed such that the state space of the automaton stores in its memory, the period and the initial block of the semi-linear sets. The idea is that this automaton's runs would be in bijection with the runs in the net $\Nc$ that take only $\gut$ transitions.  However, the counter values are `scaled down' to only remember how many periods have passed, while counter value 0 indicates that the counter value in the original run would have been at most $I+P$. The exact value of the counter value in a run of $\Nc$ can be inferred as a function of the state space. Formally, $\Bc = (Q',\Sigma,\Delta', q_0', F')$, where the set of states $Q'$  contains two types of states. One which encodes the initial block along with the current state and the other which encodes the information corresponding to the repeating block.

More formally, it is given by  $$Q' = \{\seq{q,m} \mid q\in Q \ \text{and} \ 0 \leq m \leq I   \} \cup \{\sqr{q,n} \mid q\in Q \ \text{and} \ 1 \leq n \leq P  \}.$$ The set of transitions $\Delta'$ is the union of the following sets: 
\begin{enumerate}
    \item $\{(\seq{q,i},\zerotest, a,0,\sqr{q',j}) \mid (q,i) \xrightarrow{a,d} (q',j+I) \ \text{is a} \ \gut \ \text{transition in} \ \Nc \}$
    \item $\{(\seq{q,i},\zerotest, a,0, \seq{q',j}) \mid (q,i) \xrightarrow{a,d} (q',j) \ \text{is a} \ \gut \ \text{transition in} \ \Nc\}$
    \item $\{ (\sqr{q,i},\zerotest, a,0,\seq{q,j}) \mid (q,i+I), \xrightarrow{a,d} (q',j) \ \text{is a} \ \gut \ \text{transition in} \ \Nc \}$
    \item $\{(\sqr{q,i}, \notest, a,0, \sqr{q',j})\mid (q,i+I) \xrightarrow{a,d} (q',j+I) \ \text{is a} \ \gut \ \text{transition in} \ \Nc\}$
    \item $\{(\sqr{q,i} \notest, a,1 ,\sqr{q',j}) \mid (q,i+I) \xrightarrow{a,d}(q',j+I+P) \ \text{is a} \ \gut \ \text{transition in} \ \Nc\}$ 
    \item $\{(\sqr{q,i}, \notzerotest, a,-1, \sqr{q',j}) \mid (q,i+I+P) \xrightarrow{a,d} (q',j+L) \ \text{is a} \ \gut \ \text{transition in} \ \Nc \}$
\end{enumerate}
 Here $X$ can be any symbol in $\{ \zerotest,\notzerotest\}$. The initial state is $q_0' = \seq{q_0,0}$, and the set of final states $F'$ is given by $F' = \{\seq{q,m} \mid q\in F \ \text{and} \ 0 \leq m \leq I  \} \cup \{\sqr{q,n} \mid q\in F \ \text{and} \ 1 \leq n \leq P  \}$. We note that any run in $\Bc$ starting at $q_0' = \seq{q_0,0}$ only reaches a state $\seq{q,i}$ with counter value 0, where $q \in Q, i < I$. This is because all transitions that go to such a state $\seq{q,i}$ test for $0$. We define the set $\Cc'(\Bc) = \Cc(\Bc) \setminus \{(\seq{q,i},p) \mid p > 0, \seq{q,i} \in Q'\}$, as a subset of the configuration of $\Bc$, which we call \emph{valid configurations} of $\Bc$ . Any configuration of $\Bc$ that is not valid cannot be reached.

We show that the runs in the automaton $\Bc$ are in bijection with the runs in $\Nc$ that take only $\gut$ transitions. First, we define a bijection between the valid configuration of $\Bc$ and the configurations of $\Nc$, given by $\Psi: \Cc'(\Bc) \rightarrow \Cc(\Nc)$.  
\[ \Psi(\alpha) = \begin{cases} 
          (q,i) & \text{if}\ \alpha=(\seq{q,i},0), \ q\in Q\ \text{and} \ 0\leq i \leq I\\
          (q,i+I+c\cdot P) & \text{if}\ \alpha = (\sqr{q,i},c), \ q\in Q \ \text{and} \ 1\leq i \leq p \ \text{and} \ c\geq 0 \\
       \end{cases}
    \]
The function $\Psi$ is a bijection, as can be seen by the function $\Theta:\Cc(\Nc) \rightarrow \Cc'(\Bc)$, which is the inverse of $\Psi$.
\[ \Theta((q,i)) = \begin{cases} 
          (\seq{q,i},0) & \text{where}\ \ q\in Q\ \text{and} \ 0\leq i \leq I\\
          (\sqr{q,j},c) & \text{where}\ i = I + j+c\cdot P, \ q\in Q \ \text{and} \ 1 \leq j \leq p \ \text{and} \ c\geq 0 \\
       \end{cases}
       \]
We note that the transitions in $\Bc$ are in bijection with $\gut$ transitions in $\Nc$, as $\alpha \xrightarrow{a} \beta$ is a transition in $\Delta'$ if and only if $\Psi(\alpha) \xrightarrow{a} \Psi(\beta)$ is a $\gut$ transition in $\Nc$, by construction of $\Delta'$. Thus, we can extend this bijection to get an one-to-one correspondence between runs in $\Bc$ and runs that take only $\gut$ transitions in $\Nc$. As both $\Psi$ and $\Theta$ preserves acceptance of configurations, we get that $\Lc(\Bc) \subseteq \Lc(\Nc)$.

For any accepting word $w$, any run of $\Bc$ on the word $w$ corresponds to a run of $\Nc$ on $w$ that takes only $\gut$ transitions. By Lemma~\ref{lemma:guttransitions}, such a run on $\Nc$ must exist as the automata is history-deterministic, and it must end in an accepting state of $\Nc$, which implies the corresponding run in $\Bc$ must be accepting as well. Thus, any run of $\Bc$ on an accepting word in $\Nc$ must be an accepting run in $\Bc$, showing $\Lc(\Nc) \subseteq \Lc(\Bc)$, and hence $\Lc(\Nc) = \Lc(\Bc)$. 

Now, the deterministic one-counter automaton $\Dc$, obtained by simply deleting any minimal set of transitions from $\Bc$ to make it deterministic would accept the same language as $\Nc$. This follows from the above paragraph, as any run of $\Bc$ on an accepting word (in $\Nc$ or in $\Bc$) is accepting.  
\end{proof}

%% file: appendix4b.tex
\begin{proof}
Let $\Nc =(Q, \Sigma, \Delta, q_0,F)$, and let $\gamma = (p,a,e,p')$ be a transition in $\Delta$.
Note that we use $e$ here to denote the counter-change of $\gamma$ so as to not confuse ourselves with $d$ which we will use to denote counter-change of other transitions.

We would like to show that the set $$\Sc_{\gamma} = \{k | (p,k) \xrightarrow{a,e} (p',k+e) \ \text{is a} \ \gut \ \text{transition}\}$$ is semilinear. Note that $(p,k) \xrightarrow{a} (p',k+e)$ is a $\gut$ transition if and only if $\eve$ wins $G_1$ from $((p,k),(p,k))$ with the following restrictions in the first round of the play: 
\begin{enumerate}
    \item If $\adam$ picked $a$, then $\eve$ must pick $\gamma$ as the transition on $(p,k)$, resulting in the transition $(p,k) \xrightarrow{a,e} (p',k+e)$. If $\eve$ is unable to pick $\gamma$ (due to $k+e$ being negative), then $\eve$ loses immediately.
\end{enumerate}

We shall construct a simulation game $\simg{\Mc_{\gamma},(s,k)}{\Mc'_{\gamma},(s',k)} $ between nets $\Mc_{\gamma}$ and $\Mc'_{\gamma}$, where $s$ and $s'$ are states in $\Mc_{\gamma}$ and $\Mc'_{\gamma}$ respectively, such that $\eve$ wins the simulation game $\simg{\Mc_{\gamma},(s,k)}{\Mc'_{\gamma},(s',k)}$ if and only if $(p,k) \xrightarrow{a,e} (p',k')$ is a $\gut$ transition in $\Nc$.
The construction of $\Mc_{\gamma}$ and $\Mc'_{\gamma}$ is similar to that of the nets $\Mc$ and $\Mc'$ in Lemma ~\ref{lemma:G1Simulation}, while slightly altering the initial transitions so as to match the game $G_1$ with the above restriction for $\eve$ in the first round of the play.  

\paragraph*{Construction of $\Mc'_{\gamma}$} The net $\Mc'_{\gamma}$ is essentially the net $\Nc$, along with two additional states $s'$ and $p_{\#}$, and an additional letter $\#$. The state $s'$ has exactly one outgoing $a$-transition, $(s',a,e,p')$. Recall that $p'$ is the target state of the transition $\gamma$. This is to  capture $\eve$ only being able to take $\gamma$ on reading $a$ in first round of $G_1$. From each state $q$ which was accepting in $\Nc$, we add the transition $(q,a,0,q_{\#})$ to $q_{\#}$, and we add self loops on $q_{\#}$ at $\Sigma_{\#} = \Sigma \cup \#$. Formally, let $\Mc'_{\gamma} = (Q_{\gamma}',\Sigma_{\#},\Delta',q_0,F')$, where:
\begin{itemize}
    \item the set of states is $Q' = Q \cup \{s', p_{\#}\}$,
    \item the alphabet $\Sigma_{\#} = \Sigma \cup {\#}$,
    \item the set of accepting states $F'$ is the singleton set $\{q_{\#}\}$, and
    \item the set of transitions $\Delta'$ is the union of the following sets:\begin{enumerate}
        \item  $\Delta$, the set of transitions in $\Nc$
        \item  $\{(s',a,e,p')\}$
        \item $\{(s',b,d,q)\mid (p,b,d,q) \in \Delta, b \in \Sigma \setminus \{a\}\}$
        \item $\{ (q,\#,0,q_{\#})\mid q \in F \}$
        \item $\{(q_{\#},b,0,q_{\#})\mid b \in \Sigma_{\#}\}$.
    \end{enumerate}
\end{itemize}
\paragraph*{Construction of $\Mc_{\gamma}$}
We construct $\Mc_{\gamma}$ to contain an initial state $s$, along with approximately $(|\Sigma| + 1)$ copies of the states in $\Nc$. These copies store the previous read letter in the state space of OCN. The transitions of $\Nc$ mimic a `one-step lag' in $\Mc_{\gamma}$. On reading an alphabet, the automaton $\Mc_{\gamma}$ takes a transition in the projection to $\Nc$ in the first component, based on the letter stored in the second component. Note that this transition in the $\Nc$ component is not based on the current letter being read. However, this current letter is now stored in the second component in the destination state. Formally, $\Mc_{\gamma} = (Q_{\gamma}, \Sigma_{\#},\Delta_{\gamma},s,F_{\gamma})$, where
\begin{itemize}
    \item $Q_{\gamma} = (Q \times \Sigma) \cup \{s\} \cup (F \times \{\#\})$,
    \item $s$ is the initial state,
    \item $F_{\gamma}$, the set of final states is $F \times \{\#\}$, and
    \item the set of transitions $\Delta_{\gamma}$ is the union of the following sets:
    \begin{itemize}
        \item $\{s \xrightarrow{b,0} (p,b)\mid b \in \Sigma\}$
        \item $\{((q,b),c,d,(q',c)) \mid (q,b,d,q') \in \Delta\}$, and
        \item $\{ ((q,b),\#,d,(q',\#)) \mid (q,b,d,q') \in \Delta \ \text{and} \ q'\in F\}$
    \end{itemize}
\end{itemize}
We claim that $(p,k) \xrightarrow{a,e} (p,k+e)$ is a $\gut$ transition if and only if $\eve$ wins the simulation game $\simg{\Mc_{\gamma},(s,k)}{\Mc'_{\gamma},(s',k)}$. Note that $(p,k) \xrightarrow{a,e} (p,k+e)$ is a $\gut$ transition if and only if $\eve$ wins the game $G_1$ from $((p,k),(p,k))$ with the restriction 1 in the first round mentioned above. 

Using an argument almost identical to that of in  Lemma~\ref{lemma:G1Simulation}, we can show that $\eve$ wins the game $G_1$ with the restriction 1 in the first round if and only if $\eve$ wins the simulation game $\simg{\Mc_{\gamma},(s,k)}{\Mc'_{\gamma},(s',k)}$. As the set of such $k$'s is semilinear by Theorem ~\ref{theorem:outsource}, we get that $\Sc_{\gamma}$ is semilinear as well.
\end{proof}

%% file: Appendix4c.tex
\begin{proof}
Let $\Hc_A = (Q_A,\Sigma,\Delta_A,q^0_A,F_A)$ and $\Hc_B = (Q_B,\Sigma,\Delta_B,q^0_B,F_B)$ be two history-deterministic OCNs. Note that we can assume that $\Hc_B$ accepts at least one word that is not in $\Hc_A$. This can be done because we can always consider the following OCN $\Mc$ instead of $\Hc_B$, where  for some symbol $\Cat \notin \Sigma$, we define the net $$\Mc = (Q'_{B}, \Sigma \cup \{\Cat\}, \Delta'_{B},q^0_B,F'_{B})$$ where
\begin{itemize}
    \item the set of states $Q'_{B} = Q_B \cup \{q_{*}\}$, for $q_{*} \notin Q$
    \item the set of transitions $\Delta'_{B} = \Delta_B \cup \{(q_0 ,\Cat, 0, q_{*} )$, and
    \item the final states $F'_B = F_B \cup \{q_{*}\}$.
\end{itemize}    
The OCN $\Mc$ is history-deterministic and the language accepted by the net $\Mc$ is $\Lc(\Mc) = \Lc(\Hc_B) \cup \{\Cat\}$. Note that $\Lc(\Hc_A) \subseteq \Lc(\Hc_B)$ if and only if $\Lc(\Hc_A) \subsetneq \Lc(\Mc)$. 

Henceforth, we will only consider such history-deterministic OCNs $\Hc_A$, $\Hc_B$ where $\Hc_B$ accepts a word $\Cat$ which is not accepted by $\Hc_A$. We construct an OCN $\Nc$, which is history-deterministic if and only if the language inclusion $\Lc(\Hc_A) \subset \Lc(\Hc_B)$ holds. Let $\Nc = (Q_{N},\Sigma \cup \{\heartsuit \},\Delta_{N},q^0_{N},F_{N})$, where
\begin{itemize}
    \item the set of states $Q_N = Q_A \cup Q_B \cup \{q^0_N\}$,
    \item  the set of transitions $\Delta_{N} = \Delta_A \cup \Delta'_{B} \cup \{q^0_N ,\heartsuit,0, q^0_A\} \cup \{(q^0_N,\heartsuit,0, q^0_B)\}$,  and 
    \item the final states $F_N = F_A \cup F_B$.
\end{itemize}
Suppose $\Nc$ constructed as above is history-deterministic. Then, there is a winning strategy for $\eve$ in the letter game which on reading $\heartsuit$ from $q_N^0$, chooses the transition $q^0_N \xrightarrow{\heartsuit,0} q^0_A$ or $q^0_N \xrightarrow{\heartsuit,0} q^0_B$. By our assumption, the language of $\Hc_B$ contains a word which is not in the language of $\Hc_A$. If $\eve$ did not choose the transition  to $q^0_B$, then $\eve$ looses the letter game as $\adam$ can give as input this word not in the language of $\Hc_A$, but in $\Hc_B$. 
Therefore, any winning strategy of $\eve$ must choose the transition to the copy of $\Hc_B$ from $q^0_N$. This implies that for any word accepted by $\Hc_A$, the resolver has a strategy henceforth to produce a sequence of transitions in $\Hc_B$ that leads to an accepting state, implying that $\Lc(\Hc_A) \subset \Lc(\Hc_B)$.

If $\Lc(\Hc_A) \subset \Lc(\Hc_B)$, the resolver only needs to deal with non-determinism in the  first step. Choosing the transition $q^0_N \xrightarrow{\heartsuit,0} q^0_B$ at $q^0_N$ ensures $\eve$ wins the letter game on $\Nc$, since $\Hc_B$ is history-deterministic.

Since the obtained OCN $\Nc$ has size linear in $\Hc_A$ and $\Hc_B$, we can check history-determinism of $\Nc$ to decide whether the inclusion $\Lc(\Hc_A) \subseteq \Lc(\Hc_B)$ holds, in $\PSPACE$. 
\end{proof}

%% file: Appendix4d.tex
\begin{proof}
Let us first show that given two history-deterministic one-counter nets $\Mc$ and $\Nc$,
\begin{itemize}
    \item($\Rightarrow$) $\Nc$ simulates $\Mc$ then  $\Lc(\Mc) \subseteq \Lc(\Nc)$ and
    \item($\Leftarrow$) $\Lc(\Mc) \subseteq \Lc(\Nc)$, then $\Nc$ simulates $\Mc$.
\end{itemize}

$\Rightarrow :$ Suppose $\Nc$ simulates $\Mc$. Then, over any accepting word $w \in \Lc(\Mc)$, there is an accepting run $\rho$ of it in $\Mc$, and as $\Nc$ simulates $\Mc$, the run corresponding to $\rho$ in $\Nc$ must be accepting as well. Thus, $w \in \Lc(\Nc)$. 

$\Leftarrow:$ Suppose $\Lc(\Mc) \subseteq \Lc(\Nc)$. Then, the player $\eve$ wins the simulation game $\simg{\Mc}{\Nc}$: The player $\eve$ can simply ignore $\adam$'s run in $\Mc$, and play according to her letter game strategy in $\Nc$. If $\adam$'s token is at an accepting state at the end of any round in the game after having read $w$, then $w \in \Lc(\Mc) \subseteq \Lc(\Nc)$, which implies $\eve$'s token must be on an accepting state as well, as $\eve$ was playing according to her letter game strategy in $\Nc$. 

 Note that $\Sigma^{*}$ can be given by a one state finite state automata. Thus, the problem of universality reduces to checking for simulation between the input net $\Mc$ and a one-state finite automaton, which is in $\P$ from the results of Kucera (Lemma 2,~\cite{Kuc00}).
 \end{proof}


%% file: Appendix5.tex
\subsection{Proof of Lemma~\ref{lemma:succinctHard}}
\input{Appendix5a.tex}
\subsection{Proof of Theorem~\ref{thm:UndecOCA}}
\begin{proof} 
Consider the following problem. 
\begin{quote} 
\emph{DOCA Inclusion:} Given two deterministic one-counter automata $\Ac$ and $\Bc$, is $\Lc(\Ac) \subseteq \Lc(\Bc)$?
\end{quote} 
We reduce DOCA inclusion to the problem of deciding whether a given one-counter automaton is history-deterministic. Valiant, in Section 5.1 of his thesis~\cite{Val73} shows that the DOCA inclusion problem is undecidable~\cite{VP75}. This shows that the problem of deciding if a given OCA is history-deterministic is also undecidable. 

Note that this construction is similar to the one in the construction of Lemma~\ref{lemma:HDinclusion}. We nevertheless re-state it here, for completeness.

Given DOCA $\Ac = (Q_A,\Sigma,\Delta_A,q^0_A,F_A)$ and $\Bc = (Q_B,\Sigma,\Delta_B,q^0_B,F_B)$, consider the automaton $\Bc'$, for some symbol $\Cat \notin \Sigma$, defined as the tuple $$\Bc' = (Q'_{B}, \Sigma \cup \{\Cat\}, \Delta'_{B},q^0_B,F'_{B}),$$where for a new element $q_{*}$ not in $Q$, 
\begin{itemize}
    \item the set of states $Q'_{B} = Q_B \cup \{q_{*}\}$,
    \item the set of transitions $\Delta'_{B} = \Delta_B \cup\{(q_0,\zerotest ,\Cat, 0, q_{*} ), (q_0,\notzerotest ,\Cat, 0, q_{*} )\}$, and
    \item the final states $F'_B = F_B \cup \{q_{*}\}$.
\end{itemize}    
The automaton $\Bc'$ is deterministic, and the language accepted by $\Bc'$, is $\Lc(\Bc') = \Lc(\Bc) \cup \{\Cat\}$. Note that $\Lc(\Ac) \subseteq \Lc(\Bc)$ if and only if $\Lc(\Ac) \subsetneq \Lc(\Bc')$. 

We now describe the automaton $\Hc$, which is history-deterministic if and only if $L(A) \subset L(B')$. Let $\Hc = (Q_{H},\Sigma \cup \{\Cat,\heartsuit \},\Delta_{H},q^0_{H},F_{H})$, where
\begin{itemize}
    \item the set of states $Q_H = Q_A \cup Q_B' \cup \{q^0_H\}$,
    \item  the set of transitions $\Delta_{H} = \Delta_A \cup \Delta'_{B} \cup \{(q^0_H,\zerotest,\heartsuit,0,q^0_A),(q^0_H, \zerotest,\heartsuit,0,q^0_B)\}$
    \item the set of final states $F_H = F_A \cup F'_B$.
\end{itemize}
Suppose $\Hc$ constructed as above is history-deterministic. Then, there is resolver on $q_H^0$, that chooses the transition $(q^0_H,\zerotest,\heartsuit,0,q^0_A)$ or $(q^0_H, \zerotest,\heartsuit,0,q^0_B)$. Note that $q^H_0$ is the only state where non-determinism occurs, on $\heartsuit$. Since the language of $\Bc'$ contains the word $\Cat$, which is not in the language of $\Ac$, the resolver must choose the transition $(q^0_H, \zerotest,\heartsuit,0,q^0_B)$, as otherwise $\eve$ loses the letter game if $\adam$ gave as input $\Cat$ after $\heartsuit$. 
Therefore, any resolver must choose the transition to the copy of $\Bc'$, on $\heartsuit$. This implies that for any word accepted by $\Ac$, the resolver has a strategy henceforth to produce a sequence of transitions in $\Bc'$, implying that this word must also be accepted by $\Bc$. Hence we have $\Lc(\Ac) \subseteq \Lc(\Bc)$.

For the other direction, suppose $\Lc(\Ac) \subseteq \Lc(\Bc)$. The resolver only needs to resolve non-determinism in the  starting state. Choosing the transition that takes it to the $\Bc'$ part of the automaton by selecting $(q^0_H, \zerotest, \heartsuit,0,q^0_B)$ at $q^0_H$ ensures $\eve$ wins the letter game and hence $\Hc$ is history-deterministic. 
\end{proof}

%% file: Appendix5a.tex
\begin{proof}
We first describe what reachability games on succinct one-counter nets (SOCN) are~\cite{Hun15,JOS18}. The arena of a SOCN-reachability game $\Gc(\Nc)$ consists of a one-counter net $\Nc = (Q,\{a\},\Delta,q_0,F)$ over an unary alphabet. However,
the states are partitioned as $Q = Q_{\lor} \uplus Q_{\land}$ among the players $\lor$ and $\land$ respectively, such that any transition is only between a $\land$ and $\lor$ state or a $\lor$ and $\land$ state. If the play is at a $\lor$ ($
\land$) state, then the player $\lor$ ($\land$) chooses a transition at that state to go to the next configuration. Moreover, these transitions are allowed to increment and decrement the counter more than 1, and can be any arbitrary value $d\in \mathbb{N}_{\geq 0}$, where $d$ is given in binary. The starting state of the game is the configuration $(q_0,0)$. We consider, for the $\lor$ player. the problem of reachability to a configuration $(q_{F},0)$ for some $q_F\in F$. 

The version of SOCN-reachability game we have defined is slightly different from the version considered in the theorem statement of Hunter~\cite{Hun15}, but one can see that these can be shown to be inter-reducible~\cite{JOS18}.

The decision problem, which we call {\scshape Socn-Game} is 
\begin{quote}
    \textit{Given:} A SOCN-reachability game $\Gc(\Nc)$ such that the counter change in transitions are encoded in binary.\\
    \textit{Question:} Does there exist a winning strategy for $\lor$ in $\Gc(\Nc)$.
\end{quote}
The above problem is known to be $\EXPSPACE$-complete~\cite{Hun15,JOS18}, and we shall show the $\EXPSPACE$-hardness for deciding history-determinism by reducing from the above problem. 

Given a SOCN-reachability game $\Gc(\Nc)$, we construct a net $\Mc$ which is history-deterministic if and only if the $\land$ player wins the reachability game. The net $\Mc$ is designed such that a transcription of any play on the succinct one-counter game is an accepting word. 

The following introduction to the construction of $\Mc$ is best read referring to the rigorous construction that follows it.

\textit{Describing the states of the automaton: }
To be able to recreate the  transcript of a play of the SOCN-reachability game in the letter game of $\Mc$, the state space contains 
\begin{enumerate}
    \item A `main' copy of the states of $\Nc$, in which the game would stay if $\adam$ had a strategy to win. These are used to maintain a run on $\Nc$.
    \item A resolution copy of states of $Q_\land$ for $\eve$ to stay in, until the non-determinism chosen by $\eve$ is faithfully re-played by $\adam$.
    \item A `copy' of $\Nc$ to ensure that any transcript of a run that is encoded indeed is a real run on the one-counter automaton $\Nc$. For each state $q$ of $\Nc$, we add states in the set $Q_{\win}$, in this copy. These states are called so, because from here, $\eve$ can win the letter game, and there is no non-determinism to resolve. 
    \item A few extra states, to preserve the winner in the succinct reachability game. These ensure that if a play of the letter game in $\Mc$ corresponds to a winning transcript of $\land$ player, then $\eve$ can go to a state with no non-determinism. However, there are also states $q_{\heartsuit}$ and $q_{\clubsuit}$ which are the states that $\adam$ would be able to make $\eve$ reach if there is a winning play for the $\lor$ player, from where $\eve$ loses the letter game.
\end{enumerate}

\textit{Describing the alphabet of the automaton: }The alphabet contains an input $a_{\delta}$ for each transition from a $\lor$ state. This is to make sure $\adam$, who picks the letter in the letter game is in charge of picking the next transition for the player $\lor$ in the game. There is also a single letter $a$ which creates non-determinism such that in the letter game, $\eve$ can resolve the non-determinism by picking the next transition. 
Later on, to ensure a fair play, $\adam$ is forced to confirm this non-determinism by reading the letter $a_\delta$ corresponding to the transition that $\eve$ had chosen. If $\adam$ picks a different transition, then the play moves to the component which we call $Q_{\win}$  from which all transcriptions of sequences that are `valid' in the original automaton are accepting, and there is no non-determinism, making it winning for $\eve$ in the letter game.

There are also some special symbols used in the following way: 
\begin{itemize}
    \item $\$$ is used by $\adam$ to indicate he is at a state $q_F \in F$ with $0$ in the counter. If he reads it anywhere else, $\eve$ wins the letter game trivially. 
    \item $\heartsuit$ and $\clubsuit$ are to be read immediately after $\$$, but the states $q_{\heartsuit}$ can only read $\heartsuit$ and $q_{\clubsuit}$ can only read $\clubsuit$.
\end{itemize}

\textit{Language accepted by the automaton: }The language accepted by the automaton would be any prefix of the words of the form 
$a_{\delta_0} a a_{\delta_1}a_{\delta_2}aa_{\delta_3}a_{\delta_4}\dots aa_{\delta_k}a_{\delta_{k+1}}$ followed by a $\$$ then one of $\heartsuit$ or $\clubsuit$ such that $\delta_0\delta_1\delta_2\dots \delta_k\delta_{k+1}$ is a valid sequence of transitions in $\Nc$. Note that the above words essentially is a sequence of letters of the form $a_{\delta}$ transitions, but with $a$s read before reading a `transition'  from a $\land$ state. 

\textit{Transitions of the automaton: }Now we intuitively describe the transitions of the automaton that accept such a language above. 
They are constructed so that reading each $a_{\delta}$  increments or decrements the counter by the same amount prescribed by $\delta$.

From the `main' copy of the state belonging to $\lor$, the transitions are such that if there was a transition from this state in $\Nc$, a copy of the transition also is added. Moreover, there are specific letters that one can read to go to the next state prescribed by the transition. 
For $\land$ however, these transitions are labelled by $a$, and perform no increment or decrement. They instead take the run to a temporary copy of of the state, from which $\adam$ can read the $a_\delta$ corresponding to the delta that $\eve$ had chosen. If the letter does not correspond to the same transition that $\eve$ had picked, then the run moves to the states $Q_{\win}$, from where $\eve$ wins the letter game.

We now define $\Mc$ formally, given a game on $\Nc$.  We first  describe the  $Q_{\win}$, a sub-net, which will form a part of the main automaton $\Mc$.
This is done so as to make sure the main definition has less clutter.
\paragraph*{The $Q_{\win}$ gadget:}
We describe the set of states $Q_{\win}$ and the transitions associated with it in more detail. Recall that this part of the net is to mainly ensure that only valid transcripts of a run on $\Nc$ are the ones that are accepted. For this, we essentially take one copy of $Q$ and one more copy of $Q_\land$.
We add one more state to recognise that the transcript has ended if a $\$$ has been read. 

We therefore have,
$Q_{\win} = Q \uplus Q_{\land} \uplus \{q_{\$}\}$. We call states $q_{\win}$ with the subscript if $q$ is from the copy of $Q$ and $q_{\twin}$ for a copy of $q\in Q_{\land}$. 
There are the following transitions: 
\begin{itemize}
    \item $(q_{\win}, a_{\delta}, d,p_{\win})$ for all $q\in Q_\lor$ and $\delta = (q,d, p)$, a transition in $\Nc$;
    \item $(p_{\win}, a, 0, p_{\twin})$ for  all $p\in Q_\land$
    \item $(p_{\twin}, a_{\delta}, d, q_{\win})$ for all $p\in Q_\land$ and $\delta = (p,d,q)$, a transition in $\Nc$;
    \item $(q_{\win},\$,0,q_\$)$ for all $q\in Q$.
\end{itemize}

We describe the automaton $\Mc = (Q', \Sigma',\Delta' ,q_0,F')$ where, 
\begin{itemize}
    \item the set of states $Q' = Q\cup \{q_\delta \mid \delta \in \Delta\} \cup \{q_{\heartsuit},q_{\clubsuit}, q_{\last}\}$, and the states $Q_{\win}$
    \item the alphabet set is $\Sigma = \{\$,\heartsuit,\clubsuit\}\cup \{a_{\delta} \mid  \delta\in \Delta\}\cup \{a\}$ 
    \item the start state is $q_0$, which is copy of the start state at $Q$, and
    \item all states of $Q'$ are final states, including $Q_\win$
\end{itemize}
The set of transitions are the union of the sets of transitions below, along with those of $Q_{\win}$ and some defined from the states to $Q_{\win}$ and one transitions back from it to $q_{\last}$ to end the word.
    \begin{enumerate}
        \item $\{(q_{\lor}, a_{q},d,q)\mid (q_{\lor}, d,q)\in \Delta \text{ and }q_{\lor} \in Q_{\lor}\}$
        \item $\{(p,a,0, p_{\delta})\mid \delta = (p,d,q)\in \Delta\text{ and }p\in Q_{\land}\}$
        \item $\{(p_{\delta},a_{\delta},d, q)\mid \delta = (p,d,q)\in \Delta\text{ and }p\in Q_{\land}\}$
        \item $\{(p_{\delta},a_{\delta'},d', {q'}_{\win})\mid \delta = (p,d,q) \text{ and } \delta' = (p,d',q'),  \delta \neq \delta'\text{ and }p\in Q_{\land}\}$
        \item $\{(q_F,\$,-1,q_{\$}), (q_F, \$, 0, q_{\clubsuit}), (q_F, \$, 0, q_{\heartsuit})\mid q_F\in F\}$
        \item $\{(q,\$,0,q_{\$}) \mid q \in Q \setminus F\}$
        \item $\{(q_{\heartsuit}, \heartsuit,0,q_{\last}), (q_{\clubsuit}, \clubsuit,0,q_{\last})\}$ 
        \item From $Q_{\win}$ we have $(q_{\$}, \heartsuit, 0, q_{\last}), (q_{\$}, \clubsuit, 0, q_{\last})\}$ 
    \end{enumerate}
\end{proof}
\begin{figure}
    \centering
    \makebox[\textwidth][c]{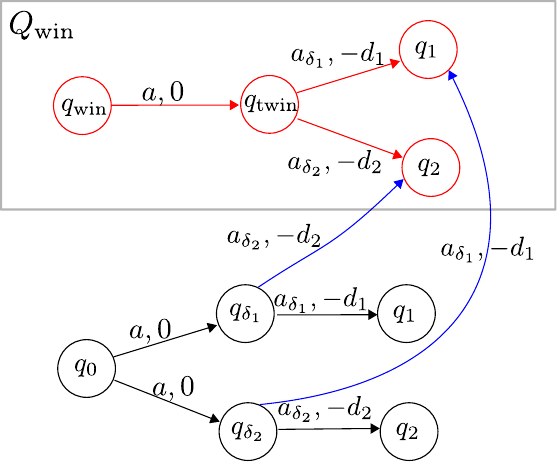}%
     \caption{A snapshot of the $Q_\win$ gadget with the states in $Q_\win$ in red, with the `main' vertices which are not in $Q_\win$ in black. If the non-determinism in the run of the succinct OCN game's imitation here is incorrectly resolved anywhere then $\eve$ can take transitions to the red states from which there is no non-determinism.\label{fig:QwinGadget}}
\end{figure} 

Observe that transitions described in items 4, 5, 6 involve transitions to $Q_{\win}$. In item 4., note that both the transitions $\delta$ and $\delta'$ should be from the same state for this transition to exist. 

We now proceed to showing that the above construction is such that
\begin{itemize}
    \item[$\Rightarrow$] If the reachability game on succinct net $\Nc$ is won by $\lor$, then $\Mc$ constructed is not history-deterministic and $\adam$ can win the letter game
    \item[$\Leftarrow$] If the reachability game on $\Nc$ is won by $\land$, then $\eve$ has a strategy to win the letter game, and $\Mc$ is history-deterministic
\end{itemize}

\paragraph*{Winning for $\lor$ reachability game implies winning for $\adam$ in letter game}
Here we prescribe $\adam$'s strategy which is essentially to follow the reachability strategy of $\lor$. 
When the game is at a state
\begin{itemize}
    \item  $q\in Q_\lor$ with counter value $k$, then he picks $a_{\delta}$ such that the transition $\delta$ ensures $(q,k)\xrightarrow{d}(q',k')$ is a winning transition prescribed by a fixed strategy in the game. This leads to no non-determinism.
    \item $q\in Q_\land$ with counter value $k$, then he picks $a$, but no matter which transition in the game $\eve$ picks, she reaches a configuration that is still winning for $\lor$ in the succinct game, this is because any $(q,k)\xrightarrow{d}(q',k')$ was a winning transition in the game.
\end{itemize}
This strategy maintains an invariant that if the play of a letter game was at a configuration such that the corresponding configuration in $\Nc$  was winning for $\lor$, then $\adam$ can ensure that in the letter game, any transition that $\eve$ picks also leads to a configuration where this is true.

Since the game is winning from $(q_0,0) \in \Cc(\Nc)$, and as $\adam$ is mimicking a winning strategy in the letter game, we know that eventually $\adam$ would reach a state $q_F$ that is in $F$ with a counter value $0$. Once he reaches such a configuration, he reads $\$$. This ensures that only the transitions $(q_F,\$,0,q_{\heartsuit})$ or $(q_F,\$,0,q_{\clubsuit})$ are enabled. From here, whichever transition $\eve$ pics, he reads the other letter corresponding to it to win.

\paragraph*{Winning for $\land$ in reachability game implies $\eve$ wins letter game}
The player $\eve$'s strategy in the letter game is to mimic the strategy of $\land$ in the underlying reachability game. Let us fix such a winning strategy for $\land$ in the reachability game. We shall show that this strategy maintains the invariant that if the play of a letter game was at a
configuration $(q,k)$ where $q\in Q$ is in the main copy of $\Nc$, then the corresponding configuration is losing for $\lor$ in the reachability game.

When the game is at a configuration $(q,k)$ in the net $\Mc$ she does the following:
\begin{itemize}
    \item for $q\in Q_\lor$ with counter value $k$, if $\adam$ picks $a_{\delta}$ such that $\delta$ ensures $(q,k)\xrightarrow{d}(q',k')$ was a transition, then $\eve$ needs to make no decisions. If not, the game proceeds to the copy $Q_{\win}$ and we can show that any sequence of runs that has a valid run is winning for the player $\eve$ anyway. 
    If the play instead stays in the `main' copy, then the new configuration reached maintains the invariant. 
    \item for $q\in Q_\land$ with counter value $k$, if $\adam$ picks $a$, then $\eve$ picks a the transitions $\delta$ corresponding to the configuration prescribed by her winning strategy. 
    Note that later if $\adam$ does not pick $a_{\delta}$, then $\eve$ wins by going to $Q_{\win}$. 
    Observe that  if the transition prescribed is such that the run goes below $0$, that run turns out to be not accepting because of the gadget described. This means $\adam$ loses the letter game again immediately. 
    If the transition still stays above 0, then the configuration proceeds to a $(q',k')$, prescribed by $\land$'s strategy in the reachability game to avoid visiting $(q_F,0)$. The position $(q',k')$ is such that there is no winning strategy for $\lor$ from it.
\end{itemize}

Observe that from any state in $Q_{\win}$, $\eve$ wins the letter game as there is no non-determinism to resolve, and any play that does not go to $Q_{\win}$ corresponds to a play in the reachability game where $\eve$'s resolution of non-determinism in the letter game corresponds to the choices of the $\land$ player in the reachability game. But note that from the invariant above, $\eve$ never reaches $(q_F,0)$ in such a play, and therefore she never has to resolve the non-determinism that occurs at $q_F$ on reading $\$$ to states $q_{\heartsuit}$ and $q_{\clubsuit}$. Any infinite play is also won by $\eve$, and therefore she wins the letter game on $\Mc$.

%% file: QwinGadget.pdf_tex
\begingroup%
  \makeatletter%
  \providecommand\color[2][]{%
    \errmessage{(Inkscape) Color is used for the text in Inkscape, but the package 'color.sty' is not loaded}%
    \renewcommand\color[2][]{}%
  }%
  \providecommand\transparent[1]{%
    \errmessage{(Inkscape) Transparency is used (non-zero) for the text in Inkscape, but the package 'transparent.sty' is not loaded}%
    \renewcommand\transparent[1]{}%
  }%
  \providecommand\rotatebox[2]{#2}%
  \newcommand*\fsize{\dimexpr\f@size pt\relax}%
  \newcommand*\lineheight[1]{\fontsize{\fsize}{#1\fsize}\selectfont}%
  \ifx\svgwidth\undefined%
    \setlength{\unitlength}{160.25415313bp}%
    \ifx\svgscale\undefined%
      \relax%
    \else%
      \setlength{\unitlength}{\unitlength * \real{\svgscale}}%
    \fi%
  \else%
    \setlength{\unitlength}{\svgwidth}%
  \fi%
  \global\let\svgwidth\undefined%
  \global\let\svgscale\undefined%
  \makeatother%
  \begin{picture}(1,0.82853726)%
    \lineheight{1}%
    \setlength\tabcolsep{0pt}%
    \put(0,0){\includegraphics[width=\unitlength,page=1]{QwinGadget.pdf}}%
  \end{picture}%
\endgroup%